\documentclass[11pt]{article}
\pdfoutput=1

\usepackage[letterpaper, portrait, margin=1in]{geometry}
\usepackage[colorlinks=true,linkcolor=blue,citecolor=ForestGreen]{hyperref}
\usepackage{algorithm}
\usepackage[noend]{algpseudocode}
\usepackage{url}
\usepackage{amsmath,amssymb,amsthm}
\usepackage{thmtools,thm-restate}
\usepackage[noabbrev,capitalise,nameinlink]{cleveref}
\usepackage{mathtools}
\usepackage{xspace}
\usepackage{verbatim}
\usepackage{mathrsfs}
\usepackage[usenames,dvipsnames,svgnames,table]{xcolor}
\usepackage{pgf}
\usepackage[dvipsnames]{xcolor}
\usepackage{todo}
\usepackage{tabularx}
\usepackage{enumitem}
\usepackage{derivative}
\usepackage{bm}
\usepackage{multirow}
\usepackage{diagbox}
\usepackage{nicematrix}

\usepackage{dsfont}

\newcommand*\ie{i.\kern.1em e.\ }
\newcommand*\eg{e.\kern.1em g.\ }
\newcommand*\cf{c.\kern.1em f.\ }

\theoremstyle{plain}
\newtheorem{theorem}{Theorem}[section]
\newtheorem{lemma}[theorem]{Lemma}
\newtheorem{fact}[theorem]{Fact}
\newtheorem{proposition}[theorem]{Proposition}

\newtheorem{corollary}[theorem]{Corollary}

\newtheorem{conjecture}[theorem]{Conjecture}
\newtheorem{observation}[theorem]{Observation}

\crefname{claim}{Claim}{Claims}

\theoremstyle{definition}

\newtheorem{definition}[theorem]{Definition}
\newtheorem{remark}[theorem]{Remark}

\theoremstyle{plain}

\newcommand{\ignore}[1]{}

\DeclareMathOperator{\poly}{poly}

\newcommand{\Var}[1]{\mathrm{Var} \left[ #1 \right]}

\newcommand{\Ex}[1]{\bE \left[ #1 \right]}
\newcommand{\Exu}[2]{\underset{#1} \bE \left[ #2 \right] }

\renewcommand{\Pr}[1]{\bP \left[ #1 \right]} 
\newcommand{\Pru}[2]{\underset{ #1 }\bP \left[ #2 \right]}

\newcommand{\define}{\vcentcolon=}

\newcommand{\floor}[1]{\ensuremath{\lfloor #1 \rfloor}}

\DeclarePairedDelimiter{\abs}{\lvert}{\rvert}

\newcommand{\ind}[1]{\mathds{1} \left[ #1 \right] }

\newcommand{\zo}{\{0,1\}}
\newcommand{\pmset}{\{\pm 1\}}

\newcommand{\cP}{\ensuremath{\mathcal{P}}}

\newcommand{\bE}{\ensuremath{\mathbb{E}}}

\newcommand{\bN}{\ensuremath{\mathbb{N}}}
\newcommand{\bP}{\ensuremath{\mathbb{P}}}
\newcommand{\bR}{\ensuremath{\mathbb{R}}}

\newcommand{\bZ}{\ensuremath{\mathbb{Z}}}

\usepackage{etoolbox}

\setlist[description]{leftmargin=\parindent,labelindent=\parindent}

\DeclarePairedDelimiterX{\inp}[2]{\langle}{\rangle}{#1, #2}

\newcommand{\grad}{\nabla}

\newcommand{\const}{\mathsf{const}}
\newcommand{\mono}{\mathsf{mono}}

\DeclareMathOperator{\Lip}{\mathsf{Lip}}

\newcommand{\lp}[2]{\left\| #1 \right\|_{L^{#2}}}

\newcommand{\lpdist}[3]{\left\| #1 \right\|_{L^{#2}\left(#3\right)}}

\algblockx{Repeat}{EndRepeat}[1]{\textbf{repeat} #1 \textbf{times}}[0]{\textbf{end repeat}}

\newtoggle{anonymous}
\togglefalse{anonymous}

\begin{document}

\title{Directed Poincaré Inequalities and $L^1$ Monotonicity Testing of Lipschitz Functions}

\iftoggle{anonymous}{%
    \author{Anonymous Author}
    }{%
\author{%
  Renato Ferreira Pinto Jr.\thanks{Partly funded by an NSERC Canada Graduate Scholarship Doctoral
  Award.}\\
  University of Waterloo\\
  \texttt{r4ferrei@uwaterloo.ca}}}

\date{}

\maketitle

\begin{abstract}
    We study the connection between directed isoperimetric inequalities and monotonicity testing. In
    recent years, this connection has unlocked breakthroughs for testing monotonicity of functions
    defined on discrete domains. Inspired the rich history of isoperimetric inequalities in
    continuous settings, we propose that studying the relationship between directed isoperimetry and
    monotonicity in such settings is essential for understanding the full scope of this connection.

    Hence, we ask whether directed isoperimetric inequalities hold for functions $f : [0,1]^n \to
    \bR$, and whether this question has implications for monotonicity testing. We answer both
    questions affirmatively. For Lipschitz functions $f : [0,1]^n \to \bR$, we show the inequality
    $d^\mono_1(f) \lesssim \Ex{\|\grad^- f\|_1}$, which upper bounds the $L^1$ distance to
    monotonicity of $f$ by a measure of its ``directed gradient''. A key ingredient in our proof is
    the \emph{monotone rearrangement} of $f$, which generalizes the classical ``sorting operator''
    to continuous settings. We use this inequality to give an $L^1$ monotonicity tester for
    Lipschitz functions $f : [0,1]^n \to \bR$, and this framework also implies similar results for
    testing real-valued functions on the hypergrid.
\end{abstract}

\thispagestyle{empty}
\setcounter{page}{0}
\newpage

\section{Introduction}
In property testing, algorithms must make a decision about whether a function $f : \Omega \to R$
has some property $\cP$, or is \emph{far} (under some distance metric) from having that property,
using a small number of queries to $f$. One of the most well-studied problems in property testing is
\emph{monotonicity testing}, the hallmark case being that of testing monotonicity of Boolean
functions on the Boolean cube, $f : \zo^n \to \zo$. We call $f$ monotone if $f(x) \le f(y)$ whenever
$x \preceq y$, \ie $x_i \le y_i$ for every $i \in [n]$.

A striking trend emerging from this topic of research has been the connection between monotonicity
testing and \emph{isoperimetric inequalities}, in particular directed analogues of classical results
such as Poincaré and Talagrand inequalities. We preview that the focus of this work is to further
explore this connection by establishing directed isoperimetric inequalities for functions $f :
[0,1]^n \to \bR$ with continuous domain and range, and as an application obtain monotonicity testers
in such settings. Before explaining our results, let us briefly summarize the connection between
monotonicity testing and directed isoperimetry.

For a function $f : \zo^n \to \bR$, let $d^\const_1(f)$ denote its $L^1$ distance to any constant
function $g : \zo^n \to \bR$, and for any point $x$, define its discrete gradient $\grad f(x) \in
\bR^n$ by $(\grad f(x))_i \define f(x^{i \to 1}) - f(x^{i \to 0})$ for each $i \in [n]$, where $x^{i
\to b}$ denotes the point $x$ with its $i$-th coordinate set to $b$. Then the following
inequality\footnote{The left-hand side is usually written $\Var{f}$ instead; for Boolean functions,
the two quantities are equivalent up to a constant factor, and writing $d^\const_1(f)$ is more
consistent with the rest of our presentation.} is usually called the Poincaré inequality on the
Boolean cube (see \eg \cite{ODon14}): for every $f : \zo^n \to \zo$,
\begin{equation}
    \label{eq:poincare-on-hypercube}
    d^\const_1(f) \lesssim \Ex{\|\grad f\|_1} \,.
\end{equation}
(Here and going forward, we write $f \lesssim g$ to denote that $f \le c g$ for some universal
constant $c$, and similarly for $f \gtrsim g$. We write $f \approx g$ to denote that $f \lesssim g$
and $g \lesssim f$.)

Now, let $d^\mono_1(f)$ denote the $L^1$ distance from $f$ to any monotone function $g : \zo^n \to
\bR$, and for each point $x$ let $\grad^- f(x)$, which we call the \emph{directed gradient} of $f$,
be given by $\grad^- f(x) \define \min\{ \grad f(x), 0 \}$. Then \cite{CS16} were the first to
notice that the main ingredient of the work of \cite{GGLRS00}, who gave a monotonicity tester for
Boolean functions on the Boolean cube with query complexity $O(n/\epsilon)$, was the following
``directed analogue'' of \eqref{eq:poincare-on-hypercube}\footnote{Typically the left-hand side
would be the distance to a \emph{Boolean} monotone function, rather than any real-valued monotone
function, but the two quantities are equal; this may be seen via a maximum matching of violating
pairs of $f$, see \cite{FLNRRS02}.}: for every
$f : \zo^n \to \zo$,
\begin{equation}
    \label{eq:edge-tester-inequality}
    d^\mono_1(f) \lesssim \Ex{\|\grad^- f\|_1} \,.
\end{equation}
The tester of \cite{GGLRS00} is the ``edge tester''\!\!, which samples edges of the Boolean cube
uniformly at random and rejects if any sampled edge violates monotonicity. Inequality
\eqref{eq:edge-tester-inequality} shows that, if $f$ is far from monotone, then many edges are
violating, so the tester stands good chance of finding one.

In their breakthrough work, \cite{CS16} gave the first monotonicity tester with $o(n)$ query
complexity by showing a directed analogue of Margulis's inequality. This was improved by
\cite{CST14}, and eventually the seminal paper of \cite{KMS18} resolved the problem of (nonadaptive)
monotonicity testing of Boolean functions on the Boolean cube, up to polylogarithmic factors, by
giving a tester with query complexity $\widetilde O(\sqrt{n} / \epsilon^2)$. The key ingredient was
to show a directed analogue of \emph{Talagrand's inequality}. Talagrand's inequality
gives that, for every $f : \zo^n \to \zo$,
\[
    d^\const_1(f) \lesssim \Ex{\|\grad f\|_2} \,.
\]
Compared to \eqref{eq:poincare-on-hypercube}, this replaces the $\ell^1$-norm of the gradient with
its $\ell^2$-norm. \cite{KMS18} showed the natural directed analogue\footnote{In fact, they require
a \emph{robust} version of this inequality, but we omit that discussion for simplicity.} up to
polylogarithmic factors, which were later removed by \cite{PRW22}: for every $f : \zo^n \to \zo$,
\[
    d^\mono_1(f) \lesssim \Ex{\|\grad^- f\|_2} \,.
\]

Since then, directed isoperimetric inequalities have also unlocked results in monotonicity testing
of Boolean functions on the hypergrid \cite{BCS18,BCS22,BKKM22,BCS23} (see also \cite{BCS20,HY22})
and real-valued functions on the Boolean cube \cite{BKR20}.

Our discussion so far has focused on isoperimetric (\emph{Poincaré-type}) inequalities on
\emph{discrete} domains. On the other hand, a rich history in geometry and functional analysis,
originated in continuous settings, has established an array of isoperimetric inequalities for
functions defined on continuous domains, as well as an impressive range of connections to topics
such as partial differential equations \cite{Poi90}, Markov diffusion processes \cite{BGL14},
probability theory and concentration of measure \cite{BL97}, optimal transport \cite{BS16},
polynomial approximation \cite{VR99}, among others. (See \cref{sec:background-inequalities} for a
brief background on Poincaré-type inequalities.)

As a motivating starting point, we note that for suitably smooth (Lipschitz) functions $f : [0,1]^n
\to \bR$, an $L^1$ Poincaré-type inequality holds \cite{BH97}:
\begin{equation}
    \label{eq:intro-poincare-inequality}
    d^\const_1(f) \lesssim \Ex{\|\grad f\|_2} \,.
\end{equation}

Thus, understanding the full scope of the connection between classical isoperimetric inequalities,
their directed counterparts, and monotonicity seems to suggest the study of the continuous setting.
In this work, we ask: do \emph{directed} Poincaré-type inequalities hold for functions $f$ with
continuous domain and range? And if so, do such inequalities have any implications for monotonicity
testing? We answer both questions affirmatively: Lipschitz functions $f : [0,1]^n \to \bR$
admit a directed $L^1$ Poincaré-type inequality (\cref{thm:main-directed-inequality-continuous}),
and this inequality implies an upper bound on the query complexity of testing monotonicity of such
functions with respect to the $L^1$ distance (\cref{thm:main-tester}). (We view $L^1$ as the natural
distance metric for the continuous setting; see \cref{sec:discussion} for a discussion.) This
framework also yields results for $L^1$ testing monotonicity of real-valued functions on the
hypergrid $f : [m]^n \to \bR$. Our testers are \emph{partial derivative testers}, which naturally
generalize the classical \emph{edge testers} \cite{GGLRS00,CS13} to continuous domains.

We now introduce our model, and then summarize our results.

\subsection{$L^p$-testing}
\label{sec:lp-testing}

Let $(\Omega, \Sigma, \mu)$ be a probability space (typically for us, the unit cube or hypergrid
with associated uniform probability distribution). Let $R \subseteq \bR$ be a range, and $\cP$ a
property of functions $g : \Omega \to R$. Given a function $f : \Omega \to \bR$, we denote the $L^p$
distance of $f$ to property $\cP$ by $d_p(f, \cP) \define \inf_{g \in \cP} d_p(f,g)$, where
$d_p(f,g) \define \Exu{x \sim \mu}{\abs*{f(x)-g(x)}^p}^{1/p}$. For fixed domain $\Omega$, we write
$d^\const_p(f)$ for the $L^p$ distance of $f$ to the property of constant functions, and
$d^\mono_p(f)$ for the $L^p$ distance of $f$ to the property of monotone functions. (See
\cref{def:lp-distance} for a formal definition contemplating \eg the required measurability and
integrability assumptions.)

\begin{definition}[$L^p$-testers]
    Let $p \ge 1$. For probability space $(\Omega,\Sigma,\mu)$, range $R \subseteq \bR$, property
    $\cP \subseteq L^p(\Omega,\mu)$ of functions $g : \Omega \to R$, and proximity parameter
    $\epsilon > 0$, we say that randomized algorithm $A$ is an \emph{$L^p$-tester for $\cP$} with
    query complexity $q$ if, given \emph{oracle access} to an unknown input function $f : \Omega \to
    R \in L^p(\Omega,\mu)$, $A$ makes at most $q$ oracle queries and 1) accepts with probability at
    least $2/3$ if $f \in \cP$; 2) rejects with probability at least $2/3$ if $d_p(f, \cP) >
    \epsilon$.
\end{definition}

We say that $A$ has \emph{one-sided error} if it accepts functions $f \in \cP$ with probability $1$,
otherwise we say it has \emph{two-sided error}. It is \emph{nonadaptive} if it decides all of its
queries in advance (\ie before seeing output from the oracle), and otherwise it is \emph{adaptive}.
We consider two types of oracle:

\begin{description}
    \item[Value oracle:] Given point $x \in \Omega$, this oracle outputs the value $f(x)$.
    \item[Directional derivative oracle:] Given point $x \in \Omega$ and vector $v \in \bR^n$, this
        oracle outputs the derivative of $f$ along $v$ at point $x$, given by
        $\frac{\partial f}{\partial v}(x) = v \cdot \grad f(x)$, as long as $f$ is differentiable at
        $x$. Otherwise, it outputs a special symbol $\bot$.
\end{description}

A directional derivative oracle is weaker than a full first-order oracle, which would return the
entire gradient \cite{BV04}, and it seems to us like a reasonable model for the
high-dimensional setting; for example, obtaining the full gradient costs $n$ queries, rather than a
single query. This type of oracle has also been studied in optimization research, \eg see
\cite{CWZ21}. For our applications, only the \emph{sign} of the result will matter, in which case we
remark that, for sufficiently smooth functions (say, functions with bounded second derivatives) each
directional derivative query may be simulated using two value queries on sufficiently close together
points.

Our definition (with value oracle) coincides with that of \cite{BRY14a} when the range is $R =
[0,1]$. On the other hand, for general $R$, we keep the distance metric unmodified, whereas
\cite{BRY14a} normalize it by the magnitude of $R$. Intuitively, we seek testers that are
efficient even when $f$ may take large values as the dimension $n$ grows; see
\cref{sec:comparison-with-prior-lp} for more details.

\subsection{Results and main ideas}
\label{section:results}

\subsubsection{Directed Poincaré-type inequalities}
\label{sec:results-inequalities}

Our first result is a directed Poincaré inequality for Lipschitz functions $f : [0,1]^n \to \bR$,
which may be seen as the continuous analogue of inequality \eqref{eq:edge-tester-inequality} of
\cite{GGLRS00}.

\begin{theorem}
    \label{thm:main-directed-inequality-continuous}
    Let $f : [0,1]^n \to \bR$ be a Lipschitz function with monotone rearrangement $f^*$. Then
    \begin{equation}
        \label{eq:main-directed-inequality}
        d^\mono_1(f) \approx \Ex{\abs*{f - f^*}} \lesssim \Ex{\|\grad^- f\|_1} \,.
    \end{equation}
\end{theorem}

As hinted in the statement, a crucial tool for this result is the \emph{monotone rearrangement}
$f^*$ of $f$. We construct $f^*$ by a sequence of axis-aligned rearrangements $R_1, \dotsc, R_n$;
each $R_i$ is the \emph{non-symmetric monotone rearrangement} operator along dimension $i$, which
naturally generalizes the \emph{sorting} operator of \cite{GGLRS00} to the continuous case. For each
coordinate $i \in [n]$, the operator $R_i$ takes $f$ into an equimeasurable function $R_i f$ that is
monotone in the $i$-th coordinate, at a ``cost'' $\Ex{\abs{f - R_if}}$ that is upper bounded by
$\Ex{\abs{\partial^-_i f}}$, where $\partial^-_i f \define (\grad^- f)_i$ is the directed partial
derivative along the $i$-th coordinate. We show that each application $R_i$ can only decrease the
``cost'' associated with further applications $R_j$, so that the total cost of obtaining $f^*$ (\ie
the LHS of \eqref{eq:main-directed-inequality}) may be upper bounded, via the triangle inequality,
by the sum of all directed partial derivatives, \ie the RHS of \eqref{eq:main-directed-inequality}.

A technically simpler version of this argument also yields a directed Poincaré inequality for
real-valued functions on the hypergrid. We also note that
\cref{thm:main-directed-inequality-continuous,thm:main-directed-inequality-discrete} are both tight
up to constant factors.

\begin{theorem}
    \label{thm:main-directed-inequality-discrete}
    Let $f : [m]^n \to \bR$ and let $f^*$ be its monotone rearrangement. Then
    \[
        d^\mono_1(f) \approx \Ex{\abs*{f - f^*}} \lesssim m \Ex{\|\grad^- f\|_1} \,.
    \]
\end{theorem}

\cref{table:inequalities} places our results in the context of existing classical and directed
inequalities. In that table and going forward, for any $p,q \ge 1$ we call the inequalities
\[
    d^\const_p(f)^p \lesssim \Ex{\|\grad f\|_q^p}
    \qquad \text{and} \qquad
    d^\mono_p(f)^p \lesssim \Ex{\|\grad^- f\|_q^p}
\]
a \emph{classical} and \emph{directed $(L^p, \ell^q)$-Poincaré inequality}, respectively. Note that
the $L^p$ notation refers to the space in which we take norms, while $\ell^q$ refers to the geometry
in which we measure gradients. In this paper, we focus on the $L^1$ inequalities. See also
\cref{sec:background-inequalities} for an extended version of \cref{table:inequalities} including
other related hypergrid inequalities shown in recent work.

We also note that we have ignored in our discussion the issues of \emph{robust} inequalities, which
seem essential for some of the testing applications (see \cite{KMS18}), and the distinction between
\emph{inner} and \emph{outer boundary}, whereby some inequalities on Boolean $f$ may be made
stronger by setting $\grad f(x)=0$ when $f(x)=0$ (see \eg \cite{Tal93}). We refer the reader to the
original works for the strongest version of each inequality and a detailed treatment of these
issues.

\begin{table}[t!]
    \centering
    \begin{NiceTabular}{c | c || c | c | c}[cell-space-limits=0.3em]
        \Block{2-2}{\diagbox{\textbf{Inequality}}{\textbf{Setting} \,}}
            & & \Block{1-2}{\textbf{Discrete}} & & \textbf{Continuous} \\
        & & $\zo^n \to \zo$ & $\zo^n \to \bR$ & $[0,1]^n \to \bR$ \\ \hline \hline

        \Block{2-1}{$(L^1, \ell^1)$-Poincaré}
            & $d^\const_1(f) \lesssim \Ex{\|\grad f\|_1}$
            & * \cite{Tal93} & * \cite{Tal93} & * \cite{BH97} \\ \cline{2-5}
        & $d^\mono_1(f) \lesssim \Ex{\|\grad^- f\|_1}$
            & \cite{GGLRS00} & \cref{thm:main-directed-inequality-discrete}
            & \cref{thm:main-directed-inequality-continuous}
            \\ \hline

        \Block{2-1}{$(L^1, \ell^2)$-Poincaré}
            & $d^\const_1(f) \lesssim \Ex{\|\grad f\|_2}$
            & * \cite{Tal93} & \cite{Tal93} & \cite{BH97} \\ \cline{2-5}
        & $d^\mono_1(f) \lesssim \Ex{\|\grad^- f\|_2}$
            & \cite{KMS18} & ? & \cref{conjecture:better-inequality}
    \end{NiceTabular}
    \caption{Classical and directed Poincaré-type inequalities on discrete and continuous domains.
    Cells marked with * indicate inequalities that follow from another entry in the table.}
    \label{table:inequalities}
\end{table}

\subsubsection{Testing monotonicity on the unit cube and hypergrid}
Equipped with the results above, we give a monotonicity tester for Lipschitz functions $f : [0,1]^n
\to \bR$, and the same technique yields a tester for functions on the hypergrid as well. The testers
are parameterized by an upper bound $L$ on the best Lipschitz constant of $f$ in $\ell^1$ geometry,
which we denote $\Lip_1(f)$ (see \cref{def:lipschitz} for a formal definition).

Both of our testers are \emph{partial derivative testers}. These are algorithms which only have
access to a directional derivative oracle and, moreover, their queries are promised to be
axis-aligned vectors. In the discrete case, these are usually called \emph{edge testers}
\cite{GGLRS00,CS13}.

\begin{theorem}
    \label{thm:main-tester}
    There is a nonadaptive partial derivative $L^1$ monotonicity tester for Lipschitz functions $f :
    [0,1]^n \to \bR$ satisfying $\Lip_1(f) \le L$ with query complexity $O\left(\frac{n
    L}{\epsilon}\right)$ and one-sided error.

    Similarly, there is a nonadaptive partial derivative $L^1$ monotonicity tester for functions $f
    : [m]^n$ satisfying $\Lip_1(f) \le L$ with query complexity $O\left(\frac{nm
    L}{\epsilon}\right)$ and one-sided error.
\end{theorem}

The testers work by sampling points $x$ and coordinates $i \in [n]$ uniformly at random, and using
directional derivative queries to reject if $\partial^-_i f(x) < 0$. Their correctness is shown
using \cref{thm:main-directed-inequality-continuous,thm:main-directed-inequality-discrete}, which
imply that, when $f$ is $\epsilon$-far from monotone in $L^1$-distance, the total magnitude of its
negative partial derivatives must be large---and since each partial derivative is at most $L$ by
assumption, the values $\partial^-_i f(x)$ must be strictly negative in a set of large measure,
which the tester stands good chance of hitting with the given query complexity.

\subsubsection{Testing monotonicity on the line}
The results above, linking a Poincaré-type inequality with a monotonicity tester that uses partial
derivative queries and has linear dependence on $n$, seem to suggest a close parallel with the case
of the edge tester on the Boolean cube \cite{GGLRS00,CS13}. On the other hand, we also show a strong
separation between Hamming and $L^1$ testing. Focusing on the simpler problem of monotonicity
testing \emph{on the line}, we show that the tight query complexity of $L^1$ monotonicity testing
Lipschitz functions grows with the square root of the size of the (continuous or discrete) domain:

\begin{theorem}
    \label{thm:tester-line}
    There exist nonadaptive $L^1$ monotonicity testers for Lipschitz functions $f : [0,m] \to \bR$
    and $f : [m] \to \bR$ satisfying $\Lip_1(f) \le L$ with query complexity $\widetilde
    O\left(\sqrt{mL/\epsilon}\right)$. The testers use value queries and have one-sided error.
\end{theorem}

This result (along with the near-tight lower bounds in \cref{sec:intro-lower-bounds}) is in contrast
with the case of Hamming testing functions $f : [m] \to \bR$, which has sample complexity
$\Theta(\log m)$ \cite{EKKRV98,Fis04,BRY14b,Bel18}. Intuitively, this difference arises because a
Lipschitz function may violate monotonicity with rate of change $L$, so the area under the curve may
grow quadratically on violating regions. The proof is in fact a reduction to the Hamming case, using
the Lipschitz assumption to establish a connection between the $L^1$ and Hamming distances to
monotonicity.

\subsubsection{Lower bounds}
\label{sec:intro-lower-bounds}

We give two types of lower bounds: under no assumptions about the tester and for constant $n$, we
show that the dependence of \cref{thm:main-tester} on $L/\epsilon$ is close to optimal\footnote{Note
that one may always multiply the input values by $1/L$ to reduce the problem to the case with
Lipschitz constant $1$ and proximity parameter $\epsilon/L$, so this is the right ratio to look
at.}. We give stronger bounds for the special case of partial derivative testers (such as the ones
from \cref{thm:main-tester}), essentially showing that our analysis of the partial derivative tester
is tight.

\begin{theorem}
    \label{thm:lower-bound-constant-n}
    Let $n$ be a constant. Any $L^1$ monotonicity tester (with two-sided error, and adaptive value
    and directional derivative queries) for Lipschitz functions $f : [0,1]^n \to \bR$ satisfying
    $\Lip_1(f) \le L$ requires at least $\Omega\left((L/\epsilon)^{\frac{n}{n+1}}\right)$ queries.

    Similarly, any $L^1$ monotonicity tester (with two-sided error and adaptive queries) for
    functions $f : [m]^n \to \bR$ satisfying $\Lip_1(f) \le L$ requires at least $\Omega\left(
    \min\left\{ (mL/\epsilon)^{\frac{n}{n+1}}, m^n \right\} \right)$ queries.
\end{theorem}

Notice that the bounds above cannot be improved beyond logarithmic factors, due to the upper bounds
for the line in \cref{thm:tester-line}. It also follows that adaptivity (essentially) does not help
with $L^1$ monotonicity testing on the line, matching the situation for Hamming testing
\cite{Fis04,CS14,Bel18}.

\cref{thm:lower-bound-constant-n} is obtained via a ``hole'' construction, which hides a
non-monotone region of $f$ inside an $\ell^1$-ball $B$ of radius $r$. We choose $r$ such the
violations of monotonicity inside $B$ are large enough to make $f$ $\epsilon$-far from monotone, but
at the same time, the ball $B$ is hard to find using few queries. However, this construction has
poor dependence on $n$.

To lower bound the query complexity of partial derivative testers with better dependence on $n$, we
employ a simpler ``step'' construction, which essentially chooses a coordinate $i$ and hides a small
negative-slope region on every line along coordinate $i$. These functions are far from monotone, but
a partial derivative tester must correctly guess both $i$ and the negative-slope region to detect
them. We conclude that \cref{thm:main-tester} is optimal for partial derivative testers on the unit
cube, and optimal for edge testers on the hypergrid for constant $\epsilon$ and $L$:

\begin{theorem}
    \label{thm:lower-bound-partial-derivative-testers}
    Any partial derivative $L^1$ monotonicity tester for Lipschitz functions $f : [0,1]^n \to \bR$
    satisfying $\Lip_1(f) \le L$ (with two-sided error and adaptive queries) requires at least
    $\Omega(nL/\epsilon)$ queries.

    For sufficiently small constant $\epsilon$ and constant $L$, any partial derivative $L^1$
    monotonicity tester for functions $f : [m]^n \to \bR$ satisfying $\Lip_1(f) \le L$ (with
    two-sided error and adaptive queries) requires at least $\Omega(nm)$ queries.
\end{theorem}

\cref{tab:testers} summarizes our upper and lower bounds for testing monotonicity on the unit cube
and hypergrid, along with the analogous Hamming testing results for intuition and bounds for $L^1$
testing from prior works. See
\cref{sec:comparison-with-prior-lp,sec:comparison,sec:bry14b-lower-bound} for a discussion and
details of how prior works imply the results in that table, since to our knowledge the problem of
$L^1$ monotonicity testing parameterized by the Lipschitz constant has not been explicitly studied
before. See also \cref{sec:related-work} for a broader overview of prior works on a spectrum of
monotonicity testing models.

\begin{table}[t!]
    \centering
    \begin{NiceTabular}{c || c || c | c}[cell-space-limits=0.3em]
        \textbf{Domain}
            & \Block{}{\textbf{Hamming testing} \\ $f : \Omega \to \bR$}
            & \Block{}{\textbf{$L^1$-testing} (prior works) \\
                $f : \Omega \to \bR$, $\Lip_1(f) \le L$}
            & \Block{}{\textbf{$L^1$-testing} (this work) \\
                $f : \Omega \to \bR$, $\Lip_1(f) \le L$} \\
                \hline \hline
        \Block{2-1}{$\Omega = [0,1]^n$}
            & \Block{2-1}{Infeasible}
            & $\widetilde O\left( \frac{n^2 L}{\epsilon} \right)$ (*) \cite{BRY14a}
            & $O\left( \frac{nL}{\epsilon} \right)$ p.d.t.
            \\ \cline{3-4}
            &
            & ---
            & \Block{}{
                \vspace{0.5em}
                $\Omega\left( \left(\frac{L}{\epsilon}\right)^{\frac{n}{n+1}} \right)$ const. $n$ \\
                $\Omega\left( \frac{nL}{\epsilon} \right)$ p.d.t.
            }
            \\ \hline \hline
        \Block{2-1}{$\Omega = [m]^n$}
            & \Block{}{
                $O\left( \frac{n\log m}{\epsilon} \right)$ \cite{CS13} \\
            }
            & \Block{}{
                $\widetilde O\left( \frac{n^2 mL}{\epsilon} \right)$ (*) \cite{BRY14a} \\
            }
            & \Block{}{
                $O\left( \frac{nmL}{\epsilon} \right)$ p.d.t. \\
            } \\ \cline{2-4}
            & \Block{}{
                $\Omega\left( \frac{n\log(m) - \log(1/\epsilon)}{\epsilon} \right)$
                    \cite{CS14} \\
            }
            & \Block{}{
                \vspace{0.5em}
                $\widetilde \Omega\left( \frac{L}{\epsilon} \right)$ n.a. 1-s. \cite{BRY14a} \\
                $\Omega(n\log m)$ n.a. \cite{BRY14b}
            }
            & \Block{}{
                \vspace{0.5em}
                $\Omega\left( \left(\frac{mL}{\epsilon}\right)^{\frac{n}{n+1}} \right)$ const. $n$
                \\
                $\Omega(nm)$ p.d.t.
            } \\
    \end{NiceTabular}
    \caption{Query complexity bounds for testing monotonicity on the unit cube and hypergrid. Upper
    bounds are for nonadaptive (n.a.) algorithms with one-sided error (1-s.), and lower bounds are
    for adaptive algorithms with two-sided error, unless stated otherwise. For $L^1$-testing, the
    upper bounds derived from prior works (*) are specialized to the Lipschitz case by us; see the
    text for details. Our lower bounds hold either for constant (const.) $n$, or for partial
    derivative testers (p.d.t.).}
    \label{tab:testers}
\end{table}

\subsection{Discussion and open questions}
\label{sec:discussion}

\subsubsection{Stronger directed Poincaré inequalities?}

Classical Poincaré inequalities are usually of the $\ell^2$ form, which seems natural \eg due to
basis independence. On the other hand, in the directed setting, the weaker $\ell^1$ inequalities (as
in \cite{GGLRS00} and
\cref{thm:main-directed-inequality-continuous,thm:main-directed-inequality-discrete}) have more
straightforward proofs than $\ell^2$ counterparts such as \cite{KMS18}. A perhaps related
observation is that monotonicity is \emph{not} a basis-independent concept, since it is defined in
terms of the standard basis. It is not obvious whether directed $\ell^2$ inequalities ought to hold
in every (real-valued, continuous) setting. Nevertheless, in light of the parallels and context
established thus far, we are hopeful that such an equality does hold. Otherwise, we believe that the
reason should be illuminating. For now, we conjecture:

\begin{conjecture}
    \label{conjecture:better-inequality}
    For every Lipschitz function $f : [0,1]^n \to \bR$, it holds that
    \[
        d^\mono_1(f) \lesssim \Ex{\|\grad^- f\|_2} \,.
    \]
\end{conjecture}

Accordingly, we also ask whether an $L^1$ tester with $O(\sqrt{n})$ complexity exists, presumably
with a dependence on the $\Lip_2(f)$ constant rather than $\Lip_1(f)$ since $\ell^2$ is the relevant
geometry above.

\subsubsection{Query complexity bounds}
Our lower bounds either have weak dependence on $n$, or only apply to a specific family of
algorithms (partial derivative testers). Previous works have established tester-independent lower
bounds with strong dependence on $n$ by using reductions from communication complexity
\cite{BBM12,BRY14b}, whose translation to the continuous setting is not obvious\footnote{Note that
there is no obvious reduction from testing on the hypergrid to testing on the unit cube---one idea
is to simulate the unit cube tester on a multilinear interpolation of the function defined on the
hypergrid, but the challenge is that simulating each query to the unit cube naively requires an
exponential number of queries to the hypergrid.}, by reduction to comparison-based testers
\cite{CS14}, whose connection to $L^1$ testing setting seems less immediate, or directly via a
careful construction \cite{Bel18}. We believe that finding strong tester-independent lower bounds
for $L^1$ testing Lipschitz functions on the unit cube is an interesting direction for further
study.

We also remark that even a tight lower bound matching \cref{thm:main-tester} may not rule out
testers with better dependence on $n$ if, for example, such a tester were parameterized by
$\Lip_2(f)$, which can be a factor of $\sqrt{n}$ larger than $\Lip_1(f)$. We view the possibility of
better testers on the unit cube, or otherwise a conceptual separation with \cite{KMS18}, as an
exciting direction for future work.

\subsubsection{Relation to prior work on $L^p$-testing}
\label{sec:comparison-with-prior-lp}

\cite{BRY14a} initiated the systematic study of $L^p$-testing and, most relevant to the present
work, established the first (and, to our knowledge, only) results on $L^p$ testing of the
monotonicity property, on the hypergrid and on the discrete line. While our models are broadly
compatible, a subtle but crucial distinction must be explained.

\cite{BRY14a} focused their exposition on the case of functions $f : \Omega \to [0,1]$, and in this
regime, $L^1$ testing can only be easier than Hamming testing, which they show via a reduction based
on Boolean threshold functions. On the other hand, for functions with other ranges, say $f : \Omega
\to [a, b]$, their definition normalizes the notion of distance by a factor of $\frac{1}{b-a}$. In
our terminology, letting $r \define b-a$ and $g \define f/r$, it follows that $d_1(g) = d_1(f)/r$,
so testing $f$ with proximity parameter $\epsilon$ reduces to testing $g$ with proximity parameter
$\epsilon/r$. For Hamming testers with query complexity that depends linearly on $1/\epsilon$, this
amounts to paying a factor of $r$ in the reduction to the Boolean case\footnote{This factor can also
be tracked explicitly in the characterization of the $L^1$ distance to monotonicity of
\cite{BRY14a}: it arises in Lemmas 2.1 and 2.2, where an integral from $0$ to $1$ must be changed to
an integral from $a$ to $b$, so the best threshold function is only guaranteed to be
$\epsilon/r$-far from monotone.}. This loss is indeed necessary, because by the same reasoning,
testing $g$ with proximity parameter $\epsilon$ reduces to testing $f$ with proximity parameter $r
\epsilon$. Therefore the problems of testing $f$ with proximity parameter $\epsilon$ and testing
$f/r$ with proximity parameter $\epsilon/r$ have the same query complexity.

In this work, we do not normalize the distance metric by $r$; we would like to handle functions $f$
that may take large values as the dimension $n$ grows, as long as $f$ satisfies a Lipschitz
assumption, and our goal is to beat the query complexity afforded by the reduction to the Boolean
case. We derive these benchmarks by assuming that the input $f$ is Lipschitz, and inferring an upper
bound on $r$ based on the Lipschitz constant and the size of the domain. Combined with the hypergrid
tester of \cite{BRY14a} and a discretization argument for the unit cube inspired by
\cite{BCS20,HY22}, we establish benchmarks for our testing problem. See \cref{sec:comparison} for
details.


With the discussion above in mind, it is instructive to return to \cref{tab:testers}. We note that
our upper bounds have polynomially smaller dependence on $n$ than the benchmarks, suggesting that
our use of the Lipschitz assumption---via the directed Poincaré inequalities in
\cref{thm:main-directed-inequality-continuous,thm:main-directed-inequality-discrete}---exploits
useful structure underlying the monotonicity testing problem (whereas the benchmark testers must
work for every function with bounded range, not only the Lipschitz ones). Our lower bounds introduce
an almost-linear dependence on the hypergrid length $m$; intuitively, this dependence is not implied
by the previous bounds in \cite{BRY14a,BRY14b} because those construct the violations of
monotonicity via Boolean functions, whereas our constructions exploit the fact that a Lipschitz
function can ``keep growing'' along a given direction, which exacerbates the $L^1$ distance to
monotonicity in the region where that happens. Our lower bounds for partial derivative testers show
that the analysis of our algorithms is essentially tight, so new (upper or lower bound) ideas are
required to establish the optimal query complexity for arbitrary testers.

\paragraph*{On the choice of $L^1$ distance and Lipschitz assumption.} We briefly motivate our
choice of distance metric and Lipschitz assumption. For continuous range and domain, well-known
counterexamples rule out testing with respect to Hamming distance: given any tester with finite
query complexity, a monotone function may be made far from monotone by arbitrarily small, hard to
detect perturbations. Testing against $L^1$ distance is then a natural choice, since this metric
takes into account the magnitude of the change required to make a function monotone (\cite{BRY14a}
also discuss connections with learning and approximation theory). However, an arbitrarily small
region of the input may still have disproportionate effect on the $L^1$ distance if the function is
arbitrary, so again testing is infeasible. Lipschitz continuity seems like a natural enough
assumption which, combined with the choice of $L^1$ distance, makes the problem tractable. Another
benefit is that Lipschitz functions are differentiable almost everywhere by Rademacher's theorem, so
the gradient is well-defined almost everywhere, which enables the connection with Poincaré-type
inequalities.

\paragraph*{Organization.} \cref{sec:prelim} introduces definitions and conventions that will be
used throughout the paper. In \cref{sec:inequalities} we prove our directed Poincaré inequalities on
the unit cube and hypergrid, and in \cref{sec:testers} we give our $L^1$ monotonicity testers for
these domains. \cref{sec:line} gives the upper bound for testing functions on the line, and in
\cref{sec:lower-bounds} we prove our lower bounds. Finally, in \cref{sec:related-work} we give a
broader overview of prior works on monotonicity testing for the reader's convenience.

\section{Preliminaries}
\label{sec:prelim}

In this paper, $\bN$ denotes the set of strictly positive integers $\{1, 2, \dotsc\}$. For $m \in
\bN$, we write $[m]$ to denote the set $\{i \in \bN : i \le m\}$. For any $c \in \bR$, we write
$c^+$ for $\max\{0, c\}$ and $c^-$ for $-\min\{0, c\}$. We denote the closure of an open set $B
\subset \bR^n$ by $\overline B$.

For a (continuous or discrete) measure space $(\Omega, \Sigma, \nu)$ and measurable function $f :
\Omega \to \bR$, we write $\int_\Omega f \odif \nu$ for the Lebesgue integral of $f$ over this
space. Then for $p \ge 1$, the space $L_p(\Omega, \nu)$ is the set of measurable functions $f$ such
that $\abs{f}^p$ is Lebesgue integrable, \ie $\int_\Omega \abs{f}^p \odif \nu < \infty$, and we
write the $L^p$ norm of such functions as $\lp{f}{p} = \lpdist{f}{p}{\nu} = \left(\int_\Omega
\abs{f}^p \odif \nu\right)^{1/p}$. We will write $\nu$ to denote the Lebesgue measure when $\Omega
\subset \bR^n$ is a continuous domain (in which case we will simply write $L^p(\Omega)$ for
$L^p(\Omega, \nu)$) and the counting measure when $\Omega \subset \bZ^n$ is a
discrete domain, and reserve $\mu$ for the special case of probability measures.

\subsection{Lipschitz functions and $L^p$ distance}

We first define Lipschitz functions with respect to a choice of $\ell^p$ geometry.

\begin{definition}
    \label{def:lipschitz}
    Let $p \ge 1$. We say that $f : \Omega \to \bR$ is \emph{$(\ell^p, L)$-Lipschitz} if,
    for every $x,y \in \Omega$, $\abs{f(x)-f(y)} \le L \|x-y\|_p$. We say that $f$ is Lipschitz if
    it is $(\ell^p, L)$-Lipschitz for any $L$ (in which case this also holds for any other choice of
    $\ell^q$), and in this case we denote by $\Lip_p(f)$ the best possible Lipschitz constant:
    \[
        \Lip_p(f) \define \inf_L \left\{ \text{$f$ is $(\ell^p, L)$-Lipschitz} \right\} \,.
    \]
    It follows that $\Lip_p(f) \le \Lip_q(f)$ for $p \le q$.
\end{definition}

We now formally define $L^p$ distances, completing the definition of $L^p$-testers from
\cref{sec:lp-testing}.

\begin{definition}[$L^p$-distance]
    \label{def:lp-distance}

    Let $p \ge 1$, let $R \subseteq \bR$, and let $(\Omega, \Sigma, \mu)$ be a probability space.
    For a property $\cP \subseteq L^p(\Omega,\mu)$ of functions $g : \Omega \to R$ and function $f :
    \Omega \to R \in L^p(\Omega,\mu)$, we define the distance from $f$ to $\cP$ as $d_p(f, \cP)
    \define \inf_{g \in \cP} d_p(f,g)$, where
    \[
        d_p(f,g) \define \lpdist{f-g}{p}{\mu}
        = \Exu{x \sim \mu}{\abs{f(x)-g(x)}^p}^{1/p} \,.
    \]
    For $p=0$, we slightly abuse notation and, taking $0^0 = 0$, write $d_0(f,g)$ for the Hamming
    distance between $f$ and $g$ weighted by $\mu$ (and $\cP$ may be any set of measurable functions
    on $(\Omega,\Sigma,\mu)$).
\end{definition}

In our applications, we will always take $\mu$ to be the uniform distribution over
$\Omega$\footnote{More precisely: when $\Omega = [0,1]^n$, $\mu$ will be the Lebesgue measure on
$\Omega$ (with associated $\sigma$-algebra $\Sigma$), and when $\Omega = [m]^n$, $\mu$ will be the
uniform distribution over $\Omega$ (with the power set of $\Omega$ as the $\sigma$-algebra
$\Sigma$).}. As a shorthand, when $(\Omega,\Sigma,\mu)$ is understood from the context and $R =
\bR$, we will write
\begin{enumerate}
    \item $d^\const_p(f) \define d_p(f, \cP^\const)$ where
        $\cP^\const \define \{ f : \Omega \to \bR \in L^p(\Omega,\mu) : f = c, c \in \bR \}$; and
    \item $d^\mono_p(f) \define d_p(f, \cP^\mono)$ where
        $\cP^\mono \define \{ f : \Omega \to \bR \in L^p(\Omega,\mu) : f \text{ is monotone} \}$.
\end{enumerate}

Going forward, we will also use the shorthand $d_p(f) \define d^\mono_p(f)$.

\subsection{Directed partial derivatives and gradients}

We first consider functions on continuous domains. Let $B$ be an open subset of $\bR^n$, and let $f
: B \to \bR$ be Lipschitz. Then by Rademacher's theorem $f$ is differentiable almost everywhere in
$B$. For each $x \in B$ where $f$ is differentiable, let $\grad f(x) = (\partial_1 f(x), \dotsc,
\partial_n f(x))$ denote its gradient, where $\partial_i f(x)$ is the partial derivative of $f$
along the $i$-th coordinate at $x$. Then, let $\partial^-_i \define \min\{0, \partial_i\}$, \ie for
every $x$ where $f$ is differentiable we have $\partial^-_i f(x) = -\left( \partial_i f(x)
\right)^-$. We call $\partial^-_i$ the \emph{directed partial derivative} operator in direction $i$.
Then we define the \emph{directed gradient} operator by $\grad^- \define (\partial^-_1, \dotsc,
\partial^-_n)$, again defined on every point $x$ where $f$ is differentiable.

Now considering the hypergrid domains, let $f : [m]^n \to \bR$. Fix $x \in [m]^n$ and $i \in [n]$,
and write $e_i$ for the $i$-th basis vector, \ie $e_i$ takes value $1$ in its $i$-th component and
$0$ elsewhere. We then define the (discrete) partial derivative of $f$ along the $i$-th coordinate
at $x$ by $\partial_i f(x) \define f(x + e_i) - f(x)$ if $x_i < m$, and $\partial_i f(x) \define 0$
if $x_i = m$. We then define its discrete gradient by $\grad \define (\partial_1, \dotsc,
\partial_n)$. Their directed counterparts are defined as above: $\partial^-_i \define \min\{0,
\partial_i\}$ and $\grad^- \define (\partial^-_1, \dotsc, \partial^-_n)$.

Note that this definition for the discrete discrete gradient on the hypergrid is slightly different
from how we introduced the discrete gradient on the Boolean cube in the opening (\cf inequality
\eqref{eq:poincare-on-hypercube}) and its use in \cref{table:inequalities}, where we allowed each
edge $(x, y)$ to ``contribute'' to both $\partial_i f(x)$ and $\partial_i f(y)$. In contrast, the
definition above (which we will use going forward) only allows the ``contribution'' to $\partial_i
f(x)$, since on domain $[m]^n$ with $m=2$, the point $y$ falls under the case $y_i = m$, so
$\partial_i f(y) \define 0$. The definition we choose seems more natural for the hypergrid settings,
but we also remark that for $\ell^1$ inequalities, the choice does not matter up to constant factors
(\ie each edge is counted once or twice). For $\ell^2$ inequalities, this choice is related to the
issues of inner/outer boundaries and robust inequalities \cite{Tal93,KMS18}.

\section{Directed Poincaré inequalities for Lipschitz functions}
\label{sec:inequalities}

In this section, we establish
\cref{thm:main-directed-inequality-continuous,thm:main-directed-inequality-discrete}. We start with
the one-dimensional case, \ie functions on the line, and then generalize to higher dimensions. In
each subsection, we will focus our presentation on the setting where the domain is continuous
(corresponding to our results for the unit cube), and then show how the same proof strategy (more
easily) yields analogous results for discrete domains (corresponding to our results for the
hypergrid).

\subsection{One-dimensional case}
Let $m > 0$, let $I \define (0,m)$, and let $f : \overline I \to \bR$ be a measurable function. We
wish to show that $\lp{f-f^*}{1} \lesssim m \lp{\partial^- f}{1}$, where $f^*$ is the monotone
rearrangement of $f$. We first introduce the monotone rearrangement, and then show this inequality
using an elementary calculus argument.

\subsubsection{Monotone rearrangement}
Here, we introduce the (non-symmetric, non-decreasing) monotone rearrangement of a one-dimensional
function. We follow the definition of \cite{Kaw85}, with the slight modification that we
are interested in the \emph{non-decreasing} rearrangement, whereas most of the literature usually
favours the non-increasing rearrangement. The difference is purely syntactic, and our choice more
conveniently matches the convention in the monotonicity testing literature. Up to this choice, our
definition also agrees with that of \cite[Chapter 2]{BS88}, and we refer the reader to these two
texts for a comprehensive treatment.

We define the (lower) \emph{level sets} of $f : \overline I \to \bR$ as the sets
\[
    \overline I_c \define \left\{ x \in \overline I : f(x) \le c \right\}
\]
for all $c \in \bR$. For nonempty measurable $S \subset \bR$ of finite measure,
the \emph{rearrangement} of $S$ is the set
\[
    S^* \define \left[ 0, \nu(S) \right]
\]
(recall that $\nu$ stands for the Lebesgue measure here),
and we define $\emptyset^* \define \emptyset$. For a level set $\overline I_c$, we write
$\overline I_c^*$ to mean $\left(\overline I_c\right)^*$.

\begin{definition}
    \label{def:monotone-rearrangement}
    The \emph{monotone rearrangement} of $f$ is the function $f^* : \overline I \to \bR$ given by
    \begin{equation}
        \label{eq:rearrangement}
        f^*(x) \define \inf\left\{ c \in \bR : x \in \overline I_c^* \right\} \,.
    \end{equation}
\end{definition}

Note that $f^*$ is always a non-decreasing function.

We note two well-known properties of the monotone rearrangement:
equimeasurability and order preservation. Two functions $f, g$ are called \emph{equimeasurable} if
$\nu\{f \ge c\} = \nu\{g \ge c\}$ for every $c \in \bR$. A mapping $u \mapsto u^*$ is called
\emph{order preserving} if $f(x) \le g(x)$ for all $x \in \overline I$ implies $f^*(x) \le g^*(x)$
for all $x \in \overline I$. See \cite[Chapter 2, Proposition 1.7]{BS88} for a proof of the
following:

\begin{fact}
    \label{prop:rearrangement-equimeasurable}
    Let $f : \overline I \to \bR$ be a measurable function. Then $f$ and $f^*$ are equimeasurable.
\end{fact}

\begin{fact}
    \label{prop:rearrangement-order-preserving}
    The mapping $f \mapsto f^*$ is order preserving.
\end{fact}

\subsubsection{Absolutely continuous functions and the one-dimensional Poincaré inequality}
Let $f : \overline I \to \bR$ be absolutely continuous. It follows that $f$ has a derivative
$\partial f$ almost everywhere (\ie outside a set of measure zero),
$\partial f \in L^1(I)$ (\ie its derivative is Lebesgue integrable), and
\[
    f(x) = f(0) + \int_0^x \partial f(t) \odif t
\]
for all $x \in \overline I$.
It also follows that $\partial^- f \in L^1(I)$.

We may now show our one-dimensional inequality:

\begin{lemma}
    \label{lemma:directed-poincare-1d}
    Let $f : \overline I \to \bR$ be absolutely continuous. Then
    $\lp{f - f^*}{1} \le 2m \lp{\partial^- f}{1}$.
\end{lemma}
\begin{proof}
    Let $S \define \left\{ x \in \overline I : f^*(x) > f(x) \right\}$,
    and note that $S$ is a measurable set because $f, f^*$ are measurable functions
    (the latter by \cref{prop:rearrangement-equimeasurable}).
    Moreover, since $f$ and $f^*$ are equimeasurable (by the same result),
    we have $\int f \odif \nu = \int f^* \odif \nu$ and therefore
    \begin{align*}
        \lp{f-f^*}{1}
        &= \int_I \abs*{f-f^*} \odif \nu
        = \int_S (f^*-f) \odif \nu + \int_{I \setminus S} (f-f^*) \odif \nu \\
        &= \int_S (f^*-f) \odif \nu +
            \left( \int_I (f-f^*) \odif \nu - \int_S (f-f^*) \odif \nu \right)
        = 2 \int_S (f^*-f) \odif \nu \,.
    \end{align*}
    Hence our goal is to show that
    \[
        \int_S (f^*-f) \odif \nu \le m \lp{\partial^- f}{1} \,.
    \]
    Let $x \in \overline I$. We claim that there exists $x' \in [0,x]$ such that $f(x') \ge f^*(x)$.
    Suppose this is not the case. Then since $f$ is continuous on $[0,x]$, by the extreme value
    theorem it attains its maximum and therefore there exists $c < f^*(x)$ such that
    $f(y) \le c$ for all $y \in [0,x]$. Thus $[0,x] \subseteq \overline I_c$, so
    $\nu\left(\overline I_c\right) \ge x$ and hence $x \in \overline I_c^*$. Then, by
    \cref{def:monotone-rearrangement}, $f^*(x) \le c < f^*(x)$, a contradiction. Thus the claim is
    proved.

    Now, let $x \in S$ and fix some $x' \in [0,x]$ such that $f(x') \ge f^*(x)$. Since $f$ is
    absolutely continuous, we have
    \[
        f^*(x) - f(x)
        \le f(x') - f(x)
        = -\int_{x'}^x \partial f(t) \odif t
        \le -\int_0^m \partial^- f(t) \odif t
        = \lp{\partial^- f}{1} \,.
    \]
    The result follows by applying this estimate to all $x$:
    \[
        \int_S (f^*-f) \odif \nu
        \le \int_S \lp{\partial^- f}{1} \odif \nu
        = \nu(S) \lp{\partial^- f}{1}
        \le m \lp{\partial^- f}{1} \,. \qedhere
    \]
\end{proof}

\subsubsection{Discrete case}
Let $m \in \bN$ and let $I \define [m]$. We may define the monotone rearrangement $f^* : I \to \bR$
of $f : I \to \bR$ as in \cref{def:monotone-rearrangement} by identifying $\overline I$ with $I$ and
writing $S^* \define [\abs{S}]$ for each finite $S \subset \bN$. More directly, $f^*$ is the
function such that $f^*(1) \le f^*(2) \le \dotsm \le f^*(m)$ is the \emph{sorted sequence} of the
values $f(1), f(2), \dotsc, f(m)$. It is easy to show that the directed version of
\cref{lemma:directed-poincare-1d} holds, and in fact one may simply repeat the proof of that lemma.

\begin{lemma}
    \label{lemma:discrete-directed-poincare-1d}
    Let $f : [m] \to \bR$. Then $\lp{f-f^*}{1} \le 2m \lp{\partial^- f}{1}$.
\end{lemma}

\subsection{Multidimensional case}

In the continuous case, we ultimately only require an inequality on the unit cube $[0,1]^n$.
However, we will first work in slightly more generality and consider functions defined on a
\emph{box} in $\bR^n$, defined below. This approach makes some of the steps more transparent, and
also gives intuition for the discrete case of the hypergrid.

\begin{definition}
    Let $a \in \bR^n_{> 0}$. The \emph{box of size $a$} is the closure $\overline B \subset \bR^n$
    of $B = (0, a_1) \times \dotsm \times (0, a_n)$.
\end{definition}

\noindent
Going forward, $\overline B \subset \bR^n$ will always denote such a box.

\paragraph*{Notation.} For $x \in \bR^n$, $y \in \bR$ and $i \in [n]$, we will use the notation
$x^{-i}$ to denote the vector in $\bR^{[n] \setminus \{i\}}$ obtained by removing the $i$-th
coordinate from $x$ (note that the indexing is not changed),
and we will write $(x^{-i}, y)$ as a shorthand for the vector
$(x_1, \dotsc, x_{i-1}, y, x_{i+1}, \dotsc, x_n) \in \bR^n$. We will also write $x^{-i}$ directly
to denote any vector in $\bR^{[n] \setminus \{i\}}$.
For function $f : \overline B \to \bR$ and $x^{-i} \in \bR^{[n] \setminus \{i\}}$,
we will write $f_{x^{-i}}$ for
the function given by $f_{x^{-i}}(y) = f(x^{-i}, y)$ for all $(x^{-i}, y) \in \overline B$.
For any set $D \in \bR^n$, we will denote by $D^{-i}$ the projection
$\{x^{-i} : x \in D\}$, and extend this notation in the natural way to more indices, \eg
$D^{-i-j}$.

\begin{definition}[Rearrangement in direction $i$]
    \label{def:rearrangement-in-direction-i}
    Let $f : \overline B \to \bR$ be a measurable function and let $i \in [n]$.
    The \emph{rearrangement of $f$ in direction $i$} is the function $R_i f : \overline B \to \bR$
    given by
    \begin{equation}
        \label{eq:rearrangement-in-direction}
        (R_i f)_{x^{-i}} \define \left( f_{x^{-i}} \right)^*
    \end{equation}
    for all $x^{-i} \in \left(\overline B\right)^{-i}$.
    We call each $R_i$ the \emph{rearrangement operator in direction $i$}.
\end{definition}

We may put \eqref{eq:rearrangement-in-direction} in words as follows: on each line in
direction $i$ determined by point $x^{-i}$, the restriction of $R_i f$ to that line is the
monotone rearrangement of the restriction of $f$ to that line.

\begin{proposition}
    \label{prop:distance-to-rearrangement-in-direction-i}
    Let $\overline B$ be the box of size $a \in \bR^n$, and let $f : \overline B \to \bR$ be Lipschitz
    continuous. Then for each $i \in [n]$,
    \[
        \lp{f - R_i f}{1} \le 2 a_i \lp{\partial_i^- f}{1} \,.
    \]
\end{proposition}
\begin{proof}
    Since $f$ is Lipschitz continuous, each $f_{x^{-i}} : [0,a_i] \to \bR$ is Lipschitz continuous
    and \emph{a fortiori} absolutely continuous. The result follows from
    \cref{lemma:directed-poincare-1d}, using Tonelli's theorem to choose the order of integration.
\end{proof}

A key ingredient in our multi-dimensional argument is that the rearrangement operator preserves
Lipschitz continuity:

\begin{lemma}[{\cite[Lemma 2.12]{Kaw85}}]
    \label{lemma:lipschitz-preserved}
    If $f : \overline B \to \bR$ is Lipschitz continuous (with Lipschitz constant $L$), then
    $R_i f$ is Lipschitz continuous (with Lipschitz constant $2L$).
\end{lemma}

We are now ready to define the (multidimensional) monotone rearrangement $f^*$:

\begin{definition}
    \label{def:monotone-rearrangement-continuous}
    Let $f : \overline B \to \bR$ be a measurable function. The \emph{monotone rearrangement}
    of $f$ is the function
    \[
        f^* \define R_n R_{n-1} \dotsm R_1 f \,.
    \]
\end{definition}

We first show that $f^*$ is indeed a monotone function:

\begin{proposition}
    \label{prop:rearrangement-is-monotone}
    Let $f : \overline B \to \bR$ be Lipschitz continuous. Then $f^*$ is monotone.
\end{proposition}
\begin{proof}
    Say that $g : \overline B \to \bR$ is \emph{monotone in direction $i$} if $g_{x^{-i}}$ is
    non-decreasing for all $x^{-i} \in \left(\overline B\right)^{-i}$. Then $g$ is monotone if and
    only if it is monotone in direction $i$ for every $i \in [n]$. Note that $R_i f$ is monotone in
    direction $i$ by definition of monotone rearrangement. Therefore, it suffices to prove that if
    $f$ is monotone in direction $j$, then $R_i f$ is also monotone in direction $j$.

    Suppose $f$ is monotone in direction $j$, and
    suppose $i < j$ without loss of generality.
    Let $a \in \bR^n$ be the size of $B$.
    Let $x^{-j} \in \left(\overline B\right)^{-j}$ and
    $0 \le y_1 < y_2 \le a_j$, so that $(x^{-j}, y_1), (x^{-j}, y_2) \in \overline B$.
    We need to show that $(R_i f)(x^{-j}, y_1) \le (R_i f)(x^{-j}, y_2)$.
    Let $\overline I_i \define [0, a_i]$.
    For each $k \in \{1,2\}$, let $g_k : \overline I_i \to \bR$ be given by
    \[
        g_k(z) \define
        f(x_1, \dotsc, x_{i-1}, z, x_{i+1}, \dotsc, x_{j-1}, y_k, x_{j+1}, \dotsc, x_n) \,.
    \]
    Note that
    \[
        g_k^*(z) =
        (R_i f)(x_1, \dotsc, x_{i-1}, z, x_{i+1}, \dotsc, x_{j-1}, y_k, x_{j+1}, \dotsc, x_n)
    \]
    for every $z \in \overline I_i$, and therefore our goal is to show that
    $g_1^*(x_i) \le g_2^*(x_i)$. But $f$ being monotone in direction $j$ means that $g_1(z) \le
    g_2(z)$ for all $z \in \overline I_i$, so by the order preserving property
    (\cref{prop:rearrangement-order-preserving}) of the monotone rearrangement we get that
    $g_1^*(x_i) \le g_2^*(x_i)$, concluding the proof.
\end{proof}

It is well-known that the monotone rearrangement is a non-expansive operator. Actually a stronger
fact holds, as we note below.

\begin{proposition}[\cite{CT80}]
    \label{proposition:rearrangement-inequalities-equivalence}
    Let $m > 0$ and let $f, g \in L^1[0,m]$. Then $f^*, g^*$ satisfy
    \[
        \int_{[0,m]} \left( f^* - g^* \right)^- \odif \nu
        \le
        \int_{[0,m]} \left( f - g \right)^- \odif \nu
    \]
    and
    \[
        \int_{[0,m]} \abs*{f^* - g^*} \odif \nu
        \le
        \int_{[0,m]} \abs*{f - g} \odif \nu \,.
    \]
\end{proposition}

The result above is stated for functions on the interval. Taking the integral over the box $B$ and
repeating for each operator $R_i$ yields the non-expansiveness of our monotone rearrangement
operator, as also noted by \cite{Kaw85}:

\begin{corollary}
    \label{cor:rearrangement-non-expansive}
    Let $f,g \in L^1(\overline B)$. Then $\lp{f^* - g^*}{1} \le \lp{f - g}{1}$.
\end{corollary}

We show that the rearrangement operator can only make the norm of the directed partial derivatives
smaller, \ie decrease the violations of monotonicity, which is the key step in this proof.

\begin{proposition}
    \label{prop:rearrangement-improves-boundary}
    Let $f : \overline B \to \bR$ be Lipschitz continuous and let $i, j \in [n]$.
    Then $\lp{\partial^-_j (R_i f)}{1} \le \lp{\partial^-_j f}{1}$.
\end{proposition}
\begin{proof}
    We may assume that $i \ne j$, since otherwise the LHS is zero.
    We will use the following convention for variables names:
    $w \in \bR^n$ will denote points in $B$;
    $z \in \bR^{[n] \setminus \{i,j\}}$ will denote points in $B^{-i-j}$;
    $x \in \bR$ will denote points in $(0, a_i)$ (indexing the $i$-th dimension); and
    $y \in \bR$ will denote points in $(0, a_j)$ (indexing the $j$-th dimension).
    For each $i \in [n]$, let $e_i$ denote the $i$-th basis vector.

    Since $f$ is Lipschitz, so is $R_i f$ by \cref{lemma:lipschitz-preserved}. By Rademacher's
    theorem, these functions are differentiable almost everywhere. Therefore, let
    $D \subseteq B$ be a measurable set such that $f$ and $R_i f$ are differentiable in $D$ and
    $\nu(D) = \nu(B)$. We have
    \begin{align*}
        \lp{\partial^-_j (R_i f)}{1}
        &= \int_D \abs*{\partial^-_j (R_i f)} \odif \nu \\
        &= \int_D \left[
            \lim_{h \to 0} \left( \frac{(R_i f)(w + h e_j) - (R_i f)(w)}{h} \right)^-
            \right] \odif \nu(w) \\
        &\overset{(BC1)}{=}
            \lim_{h \to 0}
            \int_D \left( \frac{(R_i f)(w + h e_j) - (R_i f)(w)}{h} \right)^- \odif \nu(w) \\
        &\overset{(D1)}{=}
            \lim_{h \to 0}
            \int_B \left( \frac{(R_i f)(w + h e_j) - (R_i f)(w)}{h} \right)^- \odif \nu(w) \\
        &\overset{(T1)}{=}
            \lim_{h \to 0}
            \int_{B^{-i-j}} \int_{(0,a_j)} \int_{(0,a_i)}
                \left( \frac{(R_i f)(z, y+h, x) - (R_i f)(z, y, x)}{h} \right)^-
                \odif \nu(x) \odif \nu(y) \odif \nu(z) \\
        &\le \lim_{h \to 0}
            \int_{B^{-i-j}} \int_{(0,a_j)} \int_{(0,a_i)}
                \left( \frac{f(z, y+h, x) - f(z, y, x)}{h} \right)^-
                \odif \nu(x) \odif \nu(y) \odif \nu(z) \\
        &\overset{(T2)}{=}
            \lim_{h \to 0}
            \int_B \left( \frac{f(w + h e_j) - f(w)}{h} \right)^- \odif \nu(w) \\
        &\overset{(D2)}{=}
            \lim_{h \to 0}
            \int_D \left( \frac{f(w + h e_j) - f(w)}{h} \right)^- \odif \nu(w) \\
        &\overset{(BC2)}{=}
            \int_D \left[
            \lim_{h \to 0} \left( \frac{f(w + h e_j) - f(w)}{h} \right)^-
            \right] \odif \nu(w) \\
        &= \int_D \abs*{\partial^-_j f} \odif \nu \\
        &= \lp{\partial^-_j f}{1} \,.
    \end{align*}
    Equalities (BC1) and (BC2) hold by the bounded convergence theorem, which applies because the
    difference quotients are uniformly bounded by the Lipschitz constants of $R_i f$ and $f$
    (respectively), and because $R_i f$ and $f$ are differentiable in $D$ (which gives pointwise
    convergence of the limits).
    Equalities (D1) and (D2) hold again by the uniform
    boundedness of the difference quotients, along with the fact that $\nu(B \setminus D) = 0$.
    Equalities (T1) and (T2) hold by Tonelli's theorem.
    Finally, the inequality holds by \cref{proposition:rearrangement-inequalities-equivalence},
    since $(R_i f)(z, y+h, \cdot)$ is the monotone rearrangement of $f(z, y+h, \cdot)$ and
    $(R_i f)(z, y, \cdot)$ is the monotone rearrangement of $f(z, y, \cdot)$.
\end{proof}

We are now ready to prove our directed $(L^1, \ell^1)$-Poincaré inequality.

\begin{theorem}
    \label{thm:directed-l1-isoperimetric-inequality}
    Let $B$ be the box of size $a \in \bR^n$ and let $f : \overline B \to \bR$ be Lipschitz
    continuous. Then
    \[
        \lp{f - f^*}{1} \le 2 \sum_{i=1}^n a_i \lp{\partial^-_i f}{1} \,.
    \]
\end{theorem}
\begin{proof}
    We have
    \begin{align*}
        \lp{f - f^*}{1}
        &\le \sum_{i=1}^n \lp{R_{i-1} \dotsm R_1 f - R_i \dotsm R_1 f}{1}
            & \text{(Triangle inequality)} \\
        &\le 2 \sum_{i=1}^n a_i \lp{\partial^-_i (R_{i-1} \dotsm R_1 f)}{1}
            & \text{(\cref{lemma:lipschitz-preserved,prop:distance-to-rearrangement-in-direction-i})} \\
        &\le 2 \sum_{i=1}^n a_i \lp{\partial^-_i f}{1}
            & \text{(\cref{lemma:lipschitz-preserved,prop:rearrangement-improves-boundary})} \,.
    \end{align*}
\end{proof}

Setting $B = (0,1)^n$ yields the inequality portion of
\cref{thm:main-directed-inequality-continuous}:

\begin{corollary}
    \label{cor:main-directed-inequality-continuous}
    Let $B = (0,1)^n$ and let $f : \overline B \to \bR$ be Lipschitz continuous. Then
    \[
        \Ex{\abs*{f-f^*}}
        = \lp{f - f^*}{1}
        \le 2 \int_B \| \grad^- f \|_1 \odif \nu
        = 2 \Ex{\|\grad^- f\|_1} \,.
    \]
\end{corollary}

To complete the proof of \cref{thm:main-directed-inequality-continuous}, we need to show that
$d_1(f) \approx \Ex{\abs*{f-f^*}}$, \ie that the monotone rearrangement is ``essentially optimal''
as a target monotone function for $f$. The inequality $d_1(f) \le \Ex{\abs*{f-f^*}}$ is clear from
the fact that $f^*$ is monotone. The inequality in the other direction follows from the
non-expansiveness of the rearrangement operator, with essentially the same proof as that of
\cite{KMS18} for the Boolean cube:

\begin{proposition}
    \label{prop:rearrangement-is-optimal}
    Let $f : [0,1]^n \to \bR$ be Lipschitz continuous. Then $\Ex{\abs*{f-f^*}} \le 2d_1(f)$.
\end{proposition}
\begin{proof}
    Let $g \in L^1([0,1]^n)$ be any monotone function. It follows that $g^* = g$. By
    \cref{cor:rearrangement-non-expansive}, we have that $\lp{f^*-g^*}{1} \le \lp{f-g}{1}$. Using
    the triangle inequality, we obtain
    \[
        \lp{f-f^*}{1} \le \lp{f-g}{1} + \lp{g-f^*}{1} = \lp{f-g}{1} + \lp{f^*-g^*}{1}
        \le 2\lp{f-g}{1} \,.
    \]
    The claim follows by taking the infimum over the choice of $g$.
\end{proof}

\paragraph*{Tightness of the inequality.}
To check that \cref{cor:main-directed-inequality-continuous} is tight up to constant factors, it
suffices to take the linear function $f : [0,1]^n \to \bR$ given by $f(x) = 1-x_1$ for all $x \in
[0,1]^n$. Then $f^*$ is given by $f^*(x) = x_1$, so $\Ex{f-f^*} = 1/2$ while $\Ex{\|\grad^- f\|_1} =
1$, as needed.

\subsubsection{Discrete case}

The proof above carries over to the case of the hypergrid almost unmodified, as we now outline. We
now consider functions $f : [m]^n \to \bR$, so the box $B$ is replaced with $[m]^n$ and its
dimensions $a_i$ are all replaced with the length $m$ of the hypergrid. We define the rearrangement
in direction $i$, $R_i f$, as in \cref{def:rearrangement-in-direction-i} by sorting the restrictions
of $f$ to each line along direction $i$. We also define $f^*$ as in
\cref{def:monotone-rearrangement-continuous} by subsequent applications of each operator $R_i$. Then
\cref{prop:distance-to-rearrangement-in-direction-i} carries over by applying the one-dimensional
\cref{lemma:discrete-directed-poincare-1d}, and the proof of \cref{prop:rearrangement-is-monotone}
carries over unmodified.

The non-expansiveness properties
\cref{proposition:rearrangement-inequalities-equivalence,cor:rearrangement-non-expansive} also carry
over unmodified, and the key \cref{prop:rearrangement-improves-boundary} carries over with a more
immediate proof: the use of \cref{proposition:rearrangement-inequalities-equivalence} remains the
same, but rather than expanding the definition of derivative and reasoning about the limit, the
discrete argument boils down to showing the inequality
\[
    \int_{[m]^n} \left( (R_i f)(w + e_j) - (R_i f)(w) \right)^- \odif \nu(w)
    \le \int_{[m]^n} \left( f(w + e_j) - f(w) \right)^- \odif \nu(w) \,,
\]
which follows immediately from the discrete version of
\cref{proposition:rearrangement-inequalities-equivalence} by summing over all lines in direction
$i$. Then, the hypergrid version of \cref{thm:directed-l1-isoperimetric-inequality} follows by the
same application of the triangle inequality, and we conclude the inequality portion of
\cref{thm:main-directed-inequality-discrete}:

\begin{theorem}
    \label{thm:hypergrid-inequality}
    Let $f : [m]^n \to \bR$. Then $\Ex{\abs*{f-f^*}} \le 2m \Ex{\|\grad^- f\|_1}$.
\end{theorem}

The discrete version of \cref{prop:rearrangement-is-optimal} follows identically, and we state it
here for convenience:

\begin{proposition}
    \label{prop:rearrangement-is-optimal-discrete}
    Let $f : [m]^n \to \bR$. Then $\Ex{\abs*{f-f^*}} \le 2d_1(f)$.
\end{proposition}

Finally, the tightness of \cref{thm:hypergrid-inequality} is mostly easily verified for the
following step function: letting $m$ be even for simplicity, define $f : [m]^n \to \bR$ by
\[
    f(x) = \begin{cases}
        1 & \text{if $x_1 \le m/2$} \,, \\
        0 & \text{if $x_1 > m/2$} \,.
    \end{cases}
\]
Then $f^*$ is obtained by flipping this function along the first coordinate, or equivalently
swapping the values $1$ and $0$ in the definition above. Thus $\Ex{\abs*{f-f^*}} = 1$. On the other
hand, $\|\grad^- f\|_1$ takes value $1$ on exactly one point in each line along the first
coordinate, and $0$ elsewhere. Hence $\Ex{\|\grad^- f\|_1} = 1/m$, as needed.

\section{Applications to monotonicity testing}
\label{sec:testers}

In this section, we use the directed Poincaré inequalities on the unit cube and hypergrid to show
that the natural partial derivative tester (or edge tester) attains the upper bounds from
\cref{thm:main-tester}.

Let $\Omega$ denote either $[0,1]^n$ or $[m]^n$, and let $q(\Omega, L, \epsilon)$ denote the query
complexity of testers for $(\ell^1, L)$-Lipschitz functions on these domains, as follows:
\[
    q([0,1]^n, L, \epsilon) \define \Theta\left(\frac{nL}{\epsilon}\right)
    \qquad
    \text{and}
    \qquad
    q([m]^n, L, \epsilon) \define \Theta\left(\frac{nmL}{\epsilon}\right) \,.
\]

The tester is given in \cref{alg:l1-tester-partial-derivatives}. It is clear that this algorithm is
a nonadaptive partial derivative tester, and that it always accepts monotone functions. It suffices
to show that it rejects with good probability when $d_1(f) > \epsilon$.

\begin{algorithm}[H]
    \caption{$L^1$ monotonicity tester for Lipschitz functions using partial derivative queries}
    \hspace*{\algorithmicindent}
        \textbf{Input:} Partial derivative oracle access to Lipschitz function
        $f : \Omega \to \bR$. \\
    \hspace*{\algorithmicindent}
        \textbf{Output:} Accept if $f$ is monotone, reject if $d_1(f) > \epsilon$. \\
    \hspace*{\algorithmicindent}
        \textbf{Requirement:} $\Lip_1(f) \le L$.
    \begin{algorithmic}
        \Procedure{PartialDerivativeTester}{$f, \Omega, L, \epsilon$}
            \Repeat{$q(\Omega, L, \epsilon)$}
                \State Sample $x \in \Omega$ uniformly at random.
                \State Sample $i \in [n]$ uniformly at random.
                \State \textbf{Reject} if $\partial_i f(x) < 0$.
            \EndRepeat
            \State \textbf{Accept}.
        \EndProcedure
    \end{algorithmic}
    \label{alg:l1-tester-partial-derivatives}
\end{algorithm}

\begin{lemma}
    \label{lemma:l1-tester-rejection-prob}
    Let $\Omega$ be one of $[0,1]^n$ or $[m]^n$, and let $f : \Omega \to \bR$ be a Lipschitz
    function satisfying $\Lip_1(f) \le L$. Suppose $d_1(f) > \epsilon$. Then
    \cref{alg:l1-tester-partial-derivatives} rejects with probability at least $2/3$.
\end{lemma}
\begin{proof}
    \textbf{Continuous case.} Suppose $\Omega = [0,1]^n$.
    Let $D \subseteq [0,1]^n$ be a measurable set such that $f$ is differentiable on $D$ and
    $\mu(D) = 1$, which exists by Rademacher's theorem. For each $i \in [n]$, let
    $S_i \define \{ x \in D : \partial_i f(x) < 0 \}$. A standard argument gives that
    each $S_i \subset \bR^n$ is a measurable set. We claim that
    \[
        \sum_{i=1}^n \mu(S_i) > \frac{\epsilon}{2L} \,.
    \]
    Suppose this is not the case. By the Lipschitz continuity of $f$, we have that
    $\abs*{\partial_i f(x)} \le L$ for every $x \in D$ and $i \in [n]$, and therefore
    \[
        2 \sum_{i=1}^n \Ex{\abs*{\partial^-_i f}}
        \le 2L \sum_{i=1}^n \mu(S_i)
        \le \epsilon \,.
    \]
    On the other hand, the assumption that $d_1(f) > \epsilon$ and
    \cref{cor:main-directed-inequality-continuous} yield
    \[
        \epsilon < \Ex{\abs*{f-f^*}} \le 2\Ex{\|\grad^- f\|_1}
        = 2\sum_{i=1}^n \Ex{\abs*{\partial^-_i f}} \,,
    \]
    a contradiction. Therefore the claim holds.

    Now, the probability that one iteration of the tester rejects is the probability that $x \in
    S_{i}$ when $x$ and $i$ are sampled uniformly at random. This probability is
    \[
        \Pr{\text{Iteration rejects}}
        = \sum_{j=1}^n \Pru{i}{i = j} \Pru{x}{x \in S_j}
        = \sum_{j=1}^n \frac{1}{n} \cdot \mu(S_j)
        > \frac{\epsilon}{2nL} \,.
    \]
    Thus $\Theta\left(\frac{nL}{\epsilon}\right)$ iterations suffice to reject with high
    constant probability.

    \textbf{Discrete case.} Suppose $\Omega = [m]^n$. The proof proceeds the same way, but we give
    it explicitly for convenience. For each $i \in [n]$, let $S_i \define \{ x \in
    [m]^n : \partial_i f(x) < 0\}$. We then claim that
    \[
        \sum_{i=1}^n \mu(S_i) > \frac{\epsilon}{2mL} \,.
    \]
    Indeed, if this is not the case, then since $\abs*{\partial_i f(x)} \le L$ for every $i$ and
    $x$, we get that
    \[
        2\sum_{i=1}^n \Ex{\abs*{\partial^-_i f}} \le 2L \sum_{i=1}^n \mu(S_i)
        \le \frac{\epsilon}{m} \,.
    \]
    On the other hand, the assumption that $d_1(f) > \epsilon$ and \cref{thm:hypergrid-inequality}
    yield
    \[
        \frac{\epsilon}{m} < \frac{1}{m} \cdot \Ex{\abs*{f-f^*}}
        \le \frac{1}{m} \cdot 2m \Ex{\|\grad^- f\|_1}
        = 2 \sum_{i=1}^n \Ex{\abs*{\partial^-_i f}} \,,
    \]
    a contradiction. Thus the claim holds, and the probability that one iteration of the tester
    rejects is
    \[
        \Pr{\text{Iteration rejects}}
        = \sum_{j=1}^n \Pru{i}{i=j} \Pru{x}{x \in S_j}
        = \sum_{j=1}^n \frac{1}{n} \cdot \mu(S_j)
        > \frac{\epsilon}{2nmL} \,.
    \]
    Thus $\Theta\left(\frac{nmL}{\epsilon}\right)$ iterations suffice to reject with high constant
    probability.
\end{proof}

\section{$L^1$-testing monotonicity on the line}
\label{sec:line}

In this section, we show the upper bounds for $L^1$ monotonicity testing on the line from
\cref{thm:tester-line}. The main idea is to reduce from $L^1$ testing to Hamming testing by using
the Lipschitz constant to show that, if the $L^1$ distance to monotonicity is large, then the
Hamming distance to monotonicity must be somewhat large as well; combined with the Hamming testers
of \cite{EKKRV98,Bel18}, this yields an $L^1$ tester for the discrete line $[m]$.

To obtain a tester for the continuous line $[0,m]$, we furthermore apply a discretization strategy
inspired by the domain reduction and downsampling ideas from \cite{BCS20,HY22}. The idea is that,
given $\epsilon$ and $L$, we may impose a fine enough grid on $[0,m]$ such that the function defined
on that grid preserves the $L^1$ distance to monotonicity compared to the continuous function;
again, the Lipschitz assumption is essential for this step.

In this section, we will follow the convention of denoting functions on continuous domains by $f,
g$, and those on discrete domains by $\overline f, \overline g$. Depending on the context, it will
be clear whether $\overline f$ is an arbitrary function or one obtained by discretizing a particular
function $f$. We will also write ``$f$ is $L$-Lipschitz'' without specifying the $\ell^p$ geometry,
since all choices are equivalent in one dimension.

\begin{lemma}[Discretization preserves distance to monotonicity]
    \label{lemma:continuous-to-discrete}
    Let $m, L, \epsilon > 0$ and let $f : [0,m] \to \bR$ be an $L$-Lipschitz function.
    Let the \emph{discretized function} $\overline f : [m'] \to \bR$, for suitable choice of
    $m' = \Theta\left( mL / \epsilon \right)$, be given by $\overline f(i) = f(\delta i)$
    for each $i \in [m']$, where $\delta \define m / m'$.
    Then if $d_1(f) > \epsilon$, we have $d_1(\overline f) > \epsilon/4$.
\end{lemma}
\begin{proof}
    Let $m' \in \left[ cmL/\epsilon, 2cmL/\epsilon \right]$ be an integer, where $c$ is a
    sufficiently large universal constant.\footnote{We may assume that $mL/\epsilon > 1$,
    otherwise the problem is trivial: the maximum $L^1$ distance from monotonicity attainable
    by an $L$-Lipschitz function is $\frac{1}{m} \cdot \frac{m \cdot mL}{2} = mL/2$.
    Therefore the given interval does contain an integer.} Let $\overline f : [m'] \to
    \bR$ be the function given in the statement, and suppose $d_1(f) > \epsilon$.

    Let $\overline g : [m'] \to \bR$ be the monotone rearrangement of $\overline f$.
    It is easy to check that $\overline g$ is Lipschitz with at most the Lipschitz constant of
    $\overline f$. Let $g : [0,m] \to \bR$ be the following
    piecewise linear function whose discretization is $\overline g$:
    for each $i \in [m']$ we set $g(\delta i) = \overline g(i)$, and $g$ is the linear spline
    induced by these points elsewhere (and constant in the segment $[0,\delta]$).
    Then clearly $g$ is monotone, and thus $d_1(f, g) > \epsilon$.
    Moreover, $g$ is $L$-Lipschitz,
    since its steepest slope is the same as that of $\overline g$
    up to the coordinate changes.\footnote{Formally, if $f$ is $L$-Lipschitz,
    then $\overline f$
    is $L'$-Lipschitz for $L' = Lm/m'$, hence so is its monotone rearrangement $\overline g$.
    Then since the steepest slope of $g$ must come from two vertices of the spline, $g$ is
    Lipschitz with Lipschitz constant $L'm'/m = L$.} Hence, we have
    \begin{align*}
        \epsilon &< d_1(f, g)
        = \frac{1}{m} \int_0^m \abs*{f(x) - g(x)} \odif x
        = \frac{1}{m} \sum_{i=1}^{m'} \int_{(i-1) \delta}^{i \delta} \abs*{f(x) - g(x)} \odif x \\
        &= \frac{1}{m} \sum_{i=1}^{m'} \int_{(i-1) \delta}^{i \delta}
            \abs*{\left(f(i\delta) \pm L\delta\right) - \left(g(i\delta) \pm L\delta\right)}
            \odif x
            \qquad \qquad \qquad \qquad \qquad \qquad \text{(Lipschitz property)} \\
        &\le \frac{1}{m} \sum_{i=1}^{m'} \int_{(i-1) \delta}^{i \delta} \left[
            \abs*{\overline f(i) - \overline g(i)} + 2L\delta \right] \odif x \\
        &= \frac{1}{m} \Big[
            2m'L\delta^2 + \delta \sum_{i=1}^{m'} \abs*{\overline f(i) - \overline g(i)}
            \Big]
        = \frac{2mL}{m'} + \frac{1}{m'} \sum_{i=1}^{m'} \abs*{\overline f(i) - \overline g(i)}
        \le \frac{2\epsilon}{c} + d_1(\overline f, \overline g) \,,
    \end{align*}
    where we used the notation $a \pm b$ to denote any number in the interval $[a-b, a+b]$.

    We may set $c \ge 4$ so that $2\epsilon/c \le \epsilon/2$. Therefore, we obtain
    $d_1(\overline f, \overline g) > \epsilon/2$.
    Since $\overline g$ is the monotone rearrangement of $\overline f$,
    \cref{prop:rearrangement-is-optimal-discrete} implies that
    $d_1(\overline f, \overline g) \le 2 d_1(\overline f)$. We conclude that
    $d_1(\overline f) > \epsilon/4$, as desired.
\end{proof}

\begin{observation}
    \label{obs:discrete-lipschitz}
    The function $\overline f$ defined in \cref{lemma:continuous-to-discrete} is
    $\epsilon$-Lipschitz: since $m' \ge mL/\epsilon$, we have
    \[
        \abs*{\overline f(i) - \overline f(i+1)}
        = \abs*{f(\delta i) - f\left(\delta (i+1) \right)}
        \le L \delta = Lm/m' \le \epsilon \,.
    \]
\end{observation}

\begin{lemma}[Far in $L^1$ distance imples far in Hamming distance]
    \label{lemma:l1-to-hamming}
    Let $\overline f : [m'] \to \bR$ be an $L'$-Lipschitz function.
    Then $d_0(\overline f) \ge \sqrt{\frac{d_1(\overline f)}{m'L'}}$.
\end{lemma}
\begin{proof}
    Let $S \subseteq [m']$ be a set such that 1) $|S| = d_0(\overline f) m'$ and 2) it suffices to
    change $\overline f$ on inputs in $S$ to obtain a monotone function; note that $S$ exists by
    definition of Hamming distance. Write $S$ as the union of maximal, pairwise disjoint contiguous
    intervals, $S = I_1 \cup \dotsm \cup I_k$.

    We define a monotone function $\overline g : [m'] \to \bR$ as follows. For each $i \in S$, set
    $i^* \in [m'] \setminus S$ as follows: if there exists $j \in [m'] \setminus S$ such that $j >
    i$, pick the smallest such $j$; otherwise, pick the largest $j \in [m'] \setminus S$. In other
    words, $i^*$ is obtained by picking a direction (right if possible, otherwise left) and choosing
    the first point outside the interval $I_k$ that contains $i$. Now, define $\overline g$ by
    \[
        \overline g(i) = \begin{cases}
            \overline f(i) & \text{if } i \not\in S \\
            \overline f(i^*) & \text{if } i \in S \,.
        \end{cases}
    \]

    We first claim that $\overline g$ is monotone. Indeed the sequence of values
    $\left( f(i) \right)_{i \in [m'] \setminus S}$ (taken in order of increasing $i$) is monotone
    by our first assumption on $S$, and since $\overline g$ is obtained by extending some of these
    values into flat regions, the resulting function is also monotone.
    Therefore we can upper bound the $L^1$ distance of $\overline f$ to monotonicity by
    \begin{align*}
        d_1(\overline f) \le d_1(\overline f, \overline g)
        &= \frac{1}{m'} \sum_{i=1}^{m'} \abs*{\overline f(i) - \overline g(i)}
        = \frac{1}{m'} \sum_{j=1}^k \sum_{i \in I_j} \abs*{\overline f(i) - \overline f(i^*)} \\
        &= \frac{1}{m'} \sum_{j=1}^k \sum_{i \in I_j}
            \abs*{\left( \overline f(i^*) \pm L' \abs*{i-i^*} \right) - \overline f(i^*)}
            & \text{(Lipschitz property)} \\
        &\le \frac{L'}{m'} \sum_{j=1}^k \sum_{i \in I_j} \abs*{i-i^*}
        \le \frac{L'}{m'} \sum_{j=1}^k \sum_{i \in I_j} \abs*{I_j}
        = \frac{L'}{m'} \sum_{j=1}^k \abs*{I_j}^2 \\
        &\le \frac{L'}{m'} \cdot |S|^2
            & \text{(Since $\abs*{I_1} + \dotsm + \abs*{I_k} = |S|$)} \\
        &= L' d_0(\overline f)^2 m' \,.
    \end{align*}
    The claim follows.
\end{proof}

Combining the two lemmas with the classical Hamming monotonicity tester of \cite{EKKRV98}, the
following theorem establishes \cref{thm:tester-line} for the continuous domain $[0,m]$:

\begin{theorem}
    \label{thm:tester-line-continuous}
    There exists a nonadaptive one-sided $L^1$ monotonicity tester for $L$-Lipschitz functions $f :
    [0,m] \to \bR$ with query complexity $O\left( \sqrt{\frac{mL}{\epsilon}} \log\left(
    \frac{mL}{\epsilon} \right) \right)$.
\end{theorem}
\begin{proof}
    The tester works as follows. It first fixes $m' = \Theta\left(mL/\epsilon\right)$ as given by
    \cref{lemma:continuous-to-discrete}. Let $\overline f : [m'] \to \bR$ be the discretization
    defined therein ($\overline f$ is not explicitly computed upfront, but will rather
    be queried as needed). The algorithm then simulates the (nonadaptive, one-sided)
    monotonicity tester of~\cite{EKKRV98} on
    the function $\overline f$ with proximity parameter
    $\epsilon' = \Theta\left( \sqrt{\frac{\epsilon}{mL}} \right)$ (the constant may easily be
    made explicit), producing $f(\delta i) = f(im/m')$ whenever the simulation queries
    $\overline f(i)$. The algorithm returns the result produced by the simulated tester.
    The query complexity claim follows from the fact that the tester of \cite{EKKRV98} has query
    complexity $O\left( \frac{1}{\epsilon'} \log m' \right)$.

    We now show correctness. When $f$ is monotone, so is $\overline f$, so the algorithm will
    accept since the tester of \cite{EKKRV98} has one-sided error. Now, suppose
    $d_1(f) > \epsilon$. Then $d_1(\overline f) > \epsilon/4$ by
    \cref{lemma:continuous-to-discrete}. Moreover, since $\overline f$ is $\epsilon$-Lipschitz
    by \cref{obs:discrete-lipschitz}, \cref{lemma:l1-to-hamming} implies that
    \[
        d_0(\overline f) \ge \sqrt{\frac{d_1(\overline f)}{m' \epsilon}}
        > \sqrt{\frac{1}{4 m'}} = \Omega\left( \sqrt{\frac{\epsilon}{mL}} \right) \,.
    \]
    Since this is the proximity parameter $\epsilon'$ used to instantiate the \cite{EKKRV98} tester,
    the algorithm will reject with high constant probability, as needed.
\end{proof}

\cref{lemma:l1-to-hamming} itself also implies \cref{thm:tester-line} for the discrete domain $[m]$.
This time, we use the Hamming tester of \cite{Bel18} to obtain a slightly more precise query
complexity bound\footnote{One may check that this refinement would have no effect in
\cref{thm:tester-line-continuous}}.

\begin{theorem}
    There exists a nonadaptive one-sided $L^1$ monotonicity tester for $L$-Lipschitz functions
    $\overline f : [m] \to \bR$ with query complexity
    $O\left( \sqrt{\frac{mL}{\epsilon}} \log\left( \frac{m \epsilon}{L} \right) \right)$ when
    $\epsilon/L \ge 4/m$, and $O(m)$ otherwise.
\end{theorem}
\begin{proof}
    The tester sets $\epsilon' \define \sqrt{\frac{\epsilon}{mL}}$, and then runs the (nonadaptive,
    one-sided) Hamming monotonicity tester of \cite{Bel18} on the line $[m]$ with proximity
    parameter $\epsilon'$. That tester has query complexity $O\left(\frac{1}{\epsilon'}
    \log(\epsilon' m)\right)$ when $\epsilon' \ge 2/m$ and (trivially) $O(m)$ otherwise, which gives
    the claimed upper bounds. It remains to show correctness.

    When $\overline f$ is monotone, the algorithm will accept since the tester of \cite{Bel18} has
    one-sided error. Now, suppose $d_1(\overline f) > \epsilon$. Then \cref{lemma:l1-to-hamming}
    yields
    \[
        d_0(\overline f) > \sqrt{\frac{\epsilon}{mL}} = \epsilon' \,,
    \]
    so the tester of \cite{Bel18} will reject with high constant probability.
\end{proof}

\section{Lower bounds}
\label{sec:lower-bounds}

In this section, we prove our lower bounds for testing monotonicity on the unit cube and on the
hypergrid. We first show our general lower bounds based on a ``hole'' construction, which hides a
monotonicity violating region inside a randomly placed $\ell^1$-ball; these bounds imply near
tightness of our upper bounds for testing on the line from \cref{sec:line}. Then we give our lower
bounds for partial derivative testers, which show that the analysis of our tester in
\cref{sec:testers} is tight.

\begin{definition}[$\ell^1$-ball]
    Let $\Omega$ be one of $\bR^n$ or $\bZ^n$, let $x \in \Omega$ and let $r > 0$ be a real number.
    The \emph{$\ell^1$-ball of radius $r$ centered at $x$} is the set
    $B_1^n(r, x) \define \{ y \in \Omega : \|x-y\|_1 \le r \}$. We will also write $B_1^n(r) \define
    B_1^n(r, 0)$.
\end{definition}

It will be clear from the context whether the domain should be taken to be continuous or discrete,
\ie whether $B_1^n(r,c)$ should be understood under $\Omega = \bR^n$ or $\Omega = \bZ^n$.

We give the following simple bounds on the volume of continuous and discrete $\ell^1$-balls. Since
we do not require particularly tight bounds, we opt for a simple formulation and elementary proof.

\begin{proposition}
    \label{prop:volume-l1-ball}
    There exist functions $c_1, c_2 : \bN \to \bR_{>0}$ satisfying the following. Let $n \in \bN$.
    Let $\Omega$ be one of $\bR^n$ or $\bZ^n$. Let $r \in \bR$
    satisfy $r > 0$ if $\Omega = \bR^n$, and $r \ge 1$ if $\Omega = \bZ^n$. Then
    \[
        c_1(n) r^n \le \nu\left( B_1^n(r) \right) \le c_2(n) r^n \,.
    \]
\end{proposition}
\begin{proof}
    First suppose $\Omega = \bR^n$. Then we have the following formula for the area of the
    $\ell^1$-ball of radius $r$ (see \eg \cite{Wan05}):
    \[
        \nu\left( B_1^n(r) \right) = \frac{(2r)^n}{n!} \,.
    \]
    The result follows by letting $c_1(n) \le 2^n / n!$ and $c_2(n) \ge 2^n / n!$.

    Now, suppose $\Omega = \bZ^n$, and suppose $r$ is an integer without loss of generality (because
    since $r \ge 1$, there exist integers within a factor of $2$ above and below $r$). We proceed by
    an inductive argument. For $n=1$, the volume is
    \[
        \nu\left( B_1^1 \right) = 1 + 2\sum_{d=1}^r 1 = 1 + 2r \,,
    \]
    so the claim holds by letting $c_1(1) \le 2$ and $c_2(n) \ge 3$. Assuming the claim for some $n
    \in \bN$, we have
    \begin{align*}
        \nu\left( B_1^{n+1}(r) \right)
        &= \abs*{ \left\{ x = (x_1, \dotsc, x_n, x_{n+1}): \|x\|_1 \le r \right\} }
        = \sum_{y=-r}^r \abs*{
            \left\{ x' = (x'_1, \dotsc, x'_n) : \|x'\|_1 \le r-\abs{y} \right\}
        } \\
        &= \nu\left( B_1^n(r) \right) + 2\sum_{d=1}^r \nu\left( B_1^n(r-d) \right) \,.
    \end{align*}
    Since the last expression is at most $3r \cdot \nu\left( B_1^n(r) \right)$,
    using the inductive hypothesis we conclude
    \[
        \nu\left( B_1^{n+1}(r) \right)
        \le 3r \cdot c_2(n) r^{n}
        = 3c_2(n) r^{n+1}
        \le c_2(n+1) r^{n+1} \,,
    \]
    the last inequality as long as $c_2(n+1) \ge 3c_2(n)$.

    For the lower bound, we consider two cases. Note that $r-d \ge r/2$ for at least $\floor{r/2}$
    values of $d$. When $r \ge 4$, we have $\floor{r/2} \ge r/3$, and then
    \[
        \nu\left( B_1^{n+1}(r) \right)
        \ge c_1(n) r^n + 2\sum_{d=1}^r c_1(n) (r-d)^n
        > \frac{2r}{3} \cdot c_1(n) \left(\frac{r}{2}\right)^n
        \ge c_1(n+1) r^{n+1} \,,
    \]
    the last inequality as long as $c_1(n+1) \le \frac{2}{3} \cdot 2^{-n} \cdot c_1(n)$. On the
    other hand, if $r < 4$, the bound follows easily for small enough $c_1(n+1)$, since
    \[
        \nu\left( B_1^{n+1}(r) \right)
        \ge c_1(n) r^n + 2\sum_{d=1}^r c_1(n) (r-d)^n
        > \frac{c_1(n) r^{n+1}}{r}
        > \frac{c_1(n)}{4} r^{n+1} \,. \qedhere
    \]
\end{proof}

\begin{remark}
    \label{rem:inefficient-constant}
    Note that the constants $c_1(n)$ and $c_2(n)$ in \cref{prop:volume-l1-ball} have poor dependence
    on $n$, and in particular this is tight in the continuous case. This fact is essentially the
    reason why this construction is only efficient for constant dimension $n$.
\end{remark}

We now prove our tester-independent lower bounds. Note that there exists a tester for $(\ell^1,
L)$-Lipschitz functions with proximity parameter $\epsilon$ if and only if there exists a tester for
$(\ell^1, 1)$-Lipschitz functions with proximity parameter $\epsilon/L$ (the reduction consists of
simply rescaling the input values). Therefore it suffices to prove the theorems for the case $L=1$.
The following two theorems establish the continuous and discrete cases of
\cref{thm:lower-bound-constant-n}.

\begin{theorem}[Lower bound for constant $n$ on the unit cube]
    \label{thm:lower-bound-constant-n-unit-cube}
    Let $n \in \bN$ be a constant. Any $L^1$ monotonicity tester (with two-sided error, and adaptive
    value and directional derivative queries) for Lipschitz functions $f : [0,1]^n \to \bR$
    satisfying $\Lip_1(f) \le 1$ requires at least $\Omega\left((1/\epsilon)^{\frac{n}{n+1}}\right)$
    queries.
\end{theorem}
\begin{proof}
    We construct a family of functions that are $\epsilon$-far from monotone in $L^1$ distance such
    that any deterministic algorithm cannot reliably distinguish between a function chosen uniformly
    at random from this family and the constant-$0$ function with fewer than the announced number of
    queries; then, the claim will follow from Yao's principle.

    Each such function $f$ is constructed as follows. Let $c \in [0,1]^n$ be a point such that the
    ball $B_1^n(r, c)$ is completely inside $[0,1]^n$, for radius $r$ to be chosen below. Then $f$
    takes value $0$ everywhere outside $B_1^n(r, c)$, and inside this ball, it takes value
    \[
        f(x) = -r + \|x-c\|_1
    \]
    for each $x \in B_1^n(r,c)$. Then $\Lip_1(f) = 1$. We now lower bound $d_1(f)$, its distance to
    monotonicity. Fix any $x' \in [0,1]^{n-1}$ and consider the line of points $(y, x')$ for $y \in
    [0,1]$, \ie the line along the first coordinate with remaining coordinates set to $x'$. Suppose
    this line intersects $B_1^n(r,c)$. Then this intersection occurs on some interval $[a,b]$ of
    $y$-values, and on this interval, $f$ first decreases from $f(a, x') = 0$ to
    $f\left(\frac{a+b}{2}, x'\right) = -\frac{b-a}{2}$ at rate $1$, and then increases at rate $1$
    back to $f(b, x') = 0$. Any monotone function $g$ is in particular monotone over this line, and
    it is easy to see that this requires total change to $f$ proportional to the area under this
    curve:
    \[
        \int_0^1 \abs*{f(y,x') - g(y,x')} \odif y \gtrsim \int_0^1 \abs*{f(y,x')} \odif y \,.
    \]
    Now, since this holds for any line intersecting $B_1^n(r,c)$, and the collection of such lines
    gives a partition of $B_1^n(r,c)$, the total distance between $f$ and any monotone function $g$
    is lower bounded (up to a constant) by the $L^1$-norm of $f$:
    \[
        \int_{[0,1]^n} \abs*{f-g} \odif \nu \gtrsim \int_{[0,1]^n} \abs{f} \odif \nu \,,
    \]
    and since this holds for any choice of $g$, we conclude that
    \[
        d_1(f) \gtrsim \int_{[0,1]^n} \abs{f} \odif \nu \,.
    \]
    We now note that this last expression is half the volume of an $\ell^1$-ball in dimension $n+1$:
    for each point $x \in B_1^n(r,c)$, the contribution to the integrand is $\abs*{f(x)} = r -
    \|x-c\|_1$, corresponding to the measure of points $(x,z')$ for all $0 \le z' \le z$ where $z =
    r - \|x-c\|_1$, so that the point $(x,z) \in \bR^{n+1}$ satisfies $\|(x,z) - (c,0)\|_1 = r$. In
    other words, the points $(x,z')$ are the points of $B_1^{n+1}(r,(c,0))$ with nonnegative last
    coordinate. Conversely, all such points contribute to the integral above. Therefore, since $n$
    is a constant, using \cref{prop:volume-l1-ball} and writing $\nu^{n+1}$ for the Lebesgue measure
    on $\bR^{n+1}$, we have
    \[
        d_1(f) \gtrsim \int_{[0,1]^n} \abs{f} \odif \nu
        \gtrsim \nu^{n+1}\left(B_1^{n+1}(r)\right)
        \gtrsim r^{n+1} \,.
    \]
    We wish this last quantity to be at least $\Omega(\epsilon)$, so (recalling $n$ is a constant)
    it suffices to set
    \[
        r \approx \epsilon^{\frac{1}{n+1}} \,.
    \]
    We have established that each function $f$, for this choice of $r$ and any choice of $c$, is
    $\epsilon$-far from monotone as desired. Our family of functions from which $f$ will be drawn
    will be given by choices of $c$ such that the balls $B_1^n(r,c)$ are disjoint, so that each
    query may only rule out one such choice (because queries outside $B_1^n(r,c)$ take value $0$).
    How many disjoint balls $B_1^n(r,c)$ can we fit inside $[0,1]^n$? It suffices to divide
    $[0,1]^n$ into a grid of $n$-dimensional cells of side $2r$, each of which can contain one ball.
    The number of such cells is at least (up to a constant factor)
    \[
        (1/r)^n \gtrsim (1/\epsilon)^{\frac{n}{n+1}} \,.
    \]
    Therefore to distinguish some $f$ uniformly drawn from this family from the constant-$0$
    function with constant probability, any deterministic algorithm must have query complexity at
    least $\Omega\left((1/\epsilon)^{\frac{n}{n+1}}\right)$.
\end{proof}

\begin{theorem}[Lower bound for constant $n$ on the hypergrid]
    \label{thm:lower-bound-constant-n-hypergrid}
    Let $n \in \bN$ be a constant.
    Any $L^1$ monotonicity tester (with two-sided error and adaptive queries) for functions $f :
    [m]^n \to \bR$ satisfying $\Lip_1(f) \le 1$ requires at least $\Omega\left( \min\left\{
        (m/\epsilon)^{\frac{n}{n+1}}, m^n \right\} \right)$ queries.
\end{theorem}
\begin{proof}
    We proceed similarly to \cref{thm:lower-bound-constant-n-unit-cube}, with small changes for the
    discrete setting (essentially corresponding to the requirement that $r \ge 1$ in the discrete
    case of \cref{prop:volume-l1-ball}).

    We will again construct functions $f$ based on balls $B_1^n(r,c)$ for suitable choices of $r$
    and $c$. For fixed $r$ and $c$, $f$ takes value $0$ outside the ball and, for each $x \in
    B_1^n(r,c)$,
    \[
        f(x) = -r + \|x-c\|_1 \,,
    \]
    so that $\Lip_1(f) = 1$. Again by a line restriction argument, for any monotone function $g$ we
    have
    \[
        \int_{[m]^n} \abs*{f-g} \odif \nu \gtrsim \int_{[m]^n} \abs{f} \odif \nu \,,
    \]
    and thus
    \begin{equation}
        \label{eq:integrand}
        d_1(f) \gtrsim \frac{1}{m^n} \int_{[m]^n} \abs{f} \odif \nu \,,
    \end{equation}
    the normalizing factor due to \cref{def:lp-distance}.

    When $\epsilon \le 1/m^n$, this construction boils down to setting $f(x)=-1$ at a single point
    $x$, which requires $\Omega(m^n)$ queries to identify. Now, assume $\epsilon > 1/m^n$.

    Again we may identify the integrand of \eqref{eq:integrand} with points on half of
    $B_1^{n+1}(r,(c,0))$. As long as $r \ge 1$ and since $n$ is a constant,
    \cref{prop:volume-l1-ball} implies that
    \[
        d_1(f) \gtrsim \frac{r^{n+1}}{m^n} \,.
    \]
    Thus to have $d_1(f) \ge \epsilon$, it suffices (since $n$ is a constant) to set
    \[
        r \approx m^{\frac{n}{n+1}} \epsilon^{\frac{1}{n+1}} \,,
    \]
    and indeed this gives $r \ge 1$ since $\epsilon > 1/m^n$.
    Then, our functions $f$ are given by choices of $c$ placed on the hypergrid $[m]^n$ inside
    disjoint cells of side $2r$, of which there are at least (up to a constant factor)
    \[
        \left( \frac{m}{r} \right)^n \gtrsim \left( \frac{m}{\epsilon} \right)^{\frac{n}{n+1}} \,,
    \]
    and thus any deterministic algorithm requires $\Omega\left((m/\epsilon)^{\frac{n}{n+1}}\right)$
    queries to distinguish a uniformly chosen $f$ from this family from the constant-$0$ function.
\end{proof}

The construction for the partial derivative tester lower bounds is simpler: we start with a ``step''
one-dimensional construction which is flat everywhere except for a small region of negative slope,
and then copy this function onto every line along a randomly chosen coordinate $i$. Then a
partial derivative tester must correctly guess both $i$ and the negative-slope region to detect such
functions. The following two theorems establish the continuous and discrete cases of
\cref{thm:lower-bound-partial-derivative-testers}.

\begin{theorem}[Lower bound for partial derivative testers on the unit cube]
    Any partial derivative $L^1$ monotonicity tester for Lipschitz functions $f : [0,1]^n \to \bR$
    satisfying $\Lip_1(f) \le 1$ (with two-sided error and adaptive queries) requires at least
    $\Omega(n/\epsilon)$ queries.
\end{theorem}
\begin{proof}
    Let $\epsilon \le 1/6$. For any $z \in \left[ \frac{1}{3}, \frac{2}{3}-\epsilon \right]$, let
    $g_z : [0,1] \to \bR$ be the function given by
    \[
        g(x) = \begin{cases}
            \epsilon & \text{if $x < z$} \,, \\
            \epsilon - (x-z) & \text{if $z \le x \le z + \epsilon$} \,, \\
            0 & \text{if $x > z + \epsilon$} \,. \\
        \end{cases}
    \]
    Note that $g_z$ is Lipschitz with $\Lip_1(g) = 1$. Moreover, we claim that $d_1(g_z) \gtrsim
    \epsilon$. Indeed, for any $x \in [0, 1/3]$, we have that $g_z(x) = \epsilon$ and $g_z(2/3 + x)
    = 0$. On the other hand, for any monotone function $h : [0,1] \to \bR$ we must have $h(x) \le
    h(2/3 + x)$. Thus, for any such $h$ we have $\abs*{g_z(x) - h(x)} + \abs*{g_z(2/3 + x) - h(2/3 +
    x)} \ge \epsilon$. Since this holds for all $x \in [0, 1/3]$, we conclude that for any such
    $h$ we must have $\Ex{\abs*{g_z - h}} \ge \epsilon/3$, proving the claim.

    Now, for any $i \in [n]$ and $z \in \left[ \frac{1}{3}, \frac{2}{3}-\epsilon \right]$, let
    $f_{i,z} : [0,1]^n \to \bR$ be given by copying $g_z$ onto $f$ along every line in direction
    $i$, \ie setting $f_{i,z}(x) = g_z(x_i)$ for every $x \in [0,1]^n$. Note that $\Lip_1(f) = 1$
    (since its partial derivatives are $0$ along non-$i$ coordinates), and $d_1(f) \gtrsim \epsilon$
    (since the lines in direction $i$ partition the domain).

    We construct a set of $\Omega(n/\epsilon)$ functions $f_{i,z}$ as follows. First, $i$ can be any
    of the coordinates in $[n]$. Then let $z_1, \dotsc, z_k$ be given by $z_j = \frac{1}{3} +
    k\epsilon$ for $k = \Omega(1/\epsilon)$, such that for each $j$ we have $z_j \in \left[
    \frac{1}{3}, \frac{2}{3}-\epsilon \right]$ and, moreover, for distinct $j, \ell \in [k]$,
    the regions where $f_{i,{z_j}}$ and $f_{i,{z_\ell}}$ take non-zero slope are disjoint. It
    follows that each partial derivative query may only rule out one such $f_{i,z}$, so any partial
    derivative tester that distinguishes an $f_{i,z}$ chosen uniformly at random from the
    constant-$0$ function must make at least $\Omega(n/\epsilon)$ queries.
\end{proof}

The argument for the hypergrid is similar, except that the construction cannot be made to occupy an
arbitrarily small region of the domain when the domain is discrete. We opt to keep the argument
simple and give a proof for constant parameter $\epsilon$.

\begin{theorem}[Lower bound for edge testers on the hypergrid]
    For sufficiently small constant $\epsilon$, any partial derivative $L^1$ monotonicity tester for
    functions $f : [m]^n \to \bR$ satisfying $\Lip_1(f) \le 1$ (with two-sided error and adaptive
    queries) requires at least $\Omega(nm)$ queries.
\end{theorem}
\begin{proof}
    Let $m$ be a multiple of $3$ for simplicity. For each
    $z \in \left\{ \frac{m}{3}+1, \dotsc, \frac{2m}{3} \right\}$, define $g_z : [m] \to \bR$ by
    \[
        g_z(x) = \begin{cases}
            1 & \text{if $x < z$} \,, \\
            0 & \text{if $z \ge z$} \,.
        \end{cases}
    \]
    Then $\Lip_1(g_z) = 1$ and, as before, we have $d_1(g_z) = \Omega(1)$. Then for each $i \in [n]$
    and $z \in \left[ \frac{m}{3}+1, \frac{2m}{3} \right]$, we let $f_{i,z} : [m]^n \to \bR$ be
    given by $f_{i,z}(x) = g_z(x_i)$ for each $x \in [m]^n$; it follows that $\Lip_1(f) = 1$ and
    $d_1(f) = \Omega(1)$. Note that there are $\Omega(nm)$ such functions. Moreover, each partial
    derivative query may only rule out one such $f_{i,z}$, and therefore any edge tester that
    distinguishes an $f_{i,z}$ chosen uniformly at random from the constant-$0$ function must make
    at least $\Omega(nm)$ queries.
\end{proof}

\section{Overview of prior works on monotonicity testing}
\label{sec:related-work}

We first summarize results on testing monotonicity with respect to the Hamming distance.

\vspace{-0.25em}
\paragraph*{Boolean-valued functions.} Among the early works on this problem, \cite{GGLRS00} gave
testers for functions on the hypergrid $[m]^n$ with query complexities $O(n\log(m)/\epsilon)$ and
$O((n/\epsilon)^2)$; note that the latter bound is independent of $m$, and the query complexity of
testers with this property was subsequently improved to $O((n/\epsilon) \log^2(n/\epsilon))$ by
\cite{DGLRRS99} and to $O((n/\epsilon) \log(n/\epsilon))$ by \cite{BRY14a}. For functions on the
Boolean cube $\zo^n$, \cite{CS16} gave the first $o(n)$ tester, subsequently improved by
\cite{CST14}, culminating in the $\widetilde O(\sqrt{n}/\epsilon^2)$ tester of \cite{KMS18}, which
essentially resolved the question for nonadaptive testers. Whether adaptivity helps in monotonicity
testing is still an open question; see the lower bounds below, and also \cite{CS19}.

Returning to hypergrid domains $[m]^n$, \cite{BCS18,BCS20} established first testers with $o(n)$
query complexity and, via a domain reduction technique, also obtained $o(n)$ testers for product
distributions on $\bR^n$ (and the alternative proof of \cite{HY22} improves the number of
\emph{samples} drawn by the tester when the distribution is unknown). Subsequent works
\cite{BCS22,BKKM22} attained the optimal dependence on $n$ at the cost of a dependence on $m$, with
upper bounds of the form $\widetilde O(\sqrt{n} \poly(m))$. Most recently, \cite{BCS23} gave a
tester with query complexity $O(n^{1/2+o(1)}/\epsilon^2)$, which is almost optimal for nonadaptive
algorithms, and again extends to product measures on $\bR^n$.

\vspace{-0.25em}
\paragraph*{Real-valued functions.} \cite{EKKRV98} gave a tester with query complexity
$O(\log(m)/\epsilon)$ for real-valued functions on the line $[m]$; the tight query complexity of
this problem was more recently shown to be $\Theta(\log(\epsilon m)/\epsilon)$ \cite{Bel18}. As for
functions on the hypergrid $[m]^n$, \cite{GGLRS00,DGLRRS99} also gave testers for larger ranges, but
the query complexity depends on the size of the range. Then, \cite{CS13} gave a nonadaptive tester
with one-sided error and (optimal) query complexity $O(n\log(m)/\epsilon)$.
On the Boolean cube, \cite{BKR20} gave a tester with query complexity $\widetilde O\left(
\min\left\{ r\sqrt{n}/\epsilon^2, n/\epsilon \right\} \right)$ for real-valued functions $f$ with
image size $r$, and showed that this is optimal (for constant $\epsilon$) for nonadaptive testers
with one-sided error.

\vspace{-0.25em}
\paragraph*{Lower bounds.} We briefly summarize the known lower bounds for these problems; all lower
bounds listed are for testers with two-sided error unless noted otherwise. For Boolean functions on
the Boolean cube $\zo^n$, there is a near-optimal lower bound of $\widetilde \Omega(\sqrt{n})$ for
nonadaptive testers \cite{CWX17}, which improves on prior results of \cite{FLNRRS02,CST14,CDST15}.
For adaptive testers, \cite{BB16} gave the first polynomial lower bound of $\widetilde
\Omega(n^{1/4})$, since improved to $\widetilde \Omega(n^{1/3})$ by \cite{CWX17}.

Turning to real-valued functions, \cite{Fis04} combined Ramsey theory arguments with a result of
\cite{EKKRV98} to show a $\Omega(\log m)$ lower bound for adaptive testers on the line $[m]$. On the
Boolean cube, \cite{BCGM12} gave a $\Omega(n/\epsilon)$ nonadaptive one-sided lower bound, and
\cite{BBM12} gave an adaptive lower bound of $\Omega(n)$. On the hypergrid, \cite{BRY14b} gave a
nonadaptive lower bound of $\Omega(n\log m)$ by communication complexity arguments, \cite{CS14}
showed the optimal lower bound of $\Omega(n\log(m)/\epsilon - \log(1/\epsilon)/\epsilon)$ for
adaptive testers using Ramsey theory (which involves functions with large range), and \cite{Bel18}
gave an alternative proof of this bound that does not use Ramsey theory.

\vspace{-0.25em}
\paragraph*{$L^p$-testing.} Finally, moving from Hamming testers to $L^p$ testers, and assuming
functions with range $[0,1]$, \cite{BRY14a} (who formally introduced this model) gave nonadaptive
$L^p$ monotonicity testers with one-sided error on the hypergrid $[m]^n$ with query complexity
$O((n/\epsilon^p)\log(n/\epsilon^p))$---note this is independent of $m$, bypassing the Hamming
testing lower bound---and a lower bound of $\Omega((1/\epsilon^p)\log(1/\epsilon^p))$ for
nonadaptive testers with one-sided error; on the line, they showed there is an $O(1/\epsilon^p)$
nonadaptive tester with one-sided error and a matching lower bound for adaptive testers with
two-sided error. They also gave a reduction from $L^p$ monotonicity testing to Hamming testing of
Boolean functions for nonadaptive one-sided testers, so in particular $L^1$ testing functions with
range $[0,1]$ is no harder than Hamming testing functions with Boolean range.

\vspace{-0.25em}
\paragraph*{} We also remark that our problem, which is parameterized by the upper bound $L$ on the
Lipschitz constant of input functions, lies under the umbrella of parameterized property testing,
and refer to \cite{PRV17} for an introduction to, and results on this type of tester.

\iftoggle{anonymous}{}{%
\paragraph*{Acknowledgments.} We thank Eric Blais for helpful discussions throughout the course of
this project, and for comments and suggestions on preliminary versions of this paper.
}

\bibliographystyle{alpha}
\bibliography{references}

\appendix

\section{Background on isoperimetric inequalities}
\label{sec:background-inequalities}

Let us trace an extremely brief history of developments that are most relevant to our study of
isoperimetric inequalities. We start with the original work of Poincaré \cite{Poi90}, which yields
inequalities of the type $\|f - \Ex{f} \|_p \le C(\Omega) \| \grad f \|_p$ for sufficiently smooth
domains $\Omega$ and sufficiently integrable functions $f$\footnote{More precisely, for $f$ in the
appropriate Sobolev space.}. The optimal constant $C(\Omega)$, also called the \emph{Poincaré
constant} of $\Omega$, depends on properties of this domain\footnote{For example, it is
characterized by the first nontrivial eigenvalue of the Laplacian operator on smooth bounded
$\Omega$. There are additional considerations relating the assumptions made of $f$ on the boundary
$\partial \Omega$ and the (Dirichlet or Neumann) boundary conditions associated with the Laplacian;
see \cite{KN15} for a survey.}, and often the goal is to establish the sharp constant for families
of domains $\Omega$. \cite{PW57} and \cite{AD04} (see also \cite{Beb03}) showed that, for convex
domains, $C(\Omega)$ is essentially upper bounded by the diameter of $\Omega$, and this bound is
tight in general. However, for specific structured domains such as the product domain $[0,1]^n$, the
diameter characterization falls short of yielding a dimension-free inequality (see also the
literature on logarithmic Sobolev inequalities \cite{Gro75}).

Making progress on this front in the discrete setting, the landmark work of Talagrand \cite{Tal93}
established inequalities like the above for domain $\Omega = \zo^n$, with $C = C(\Omega)$
independent of $n$, and established connections with earlier works of Margulis on graph connectivity
and Pisier on probability in Banach spaces. (More recently, Fourier-analytic proofs of Talagrand's
inequality have also been given \cite{EKLM22}.) In continuous settings, similar results were first
established for the Gaussian measure in connection with the Gaussian isoperimetric inequality
\cite{Bob97,BL96,ST78,Bor75,Lat03}. Tying back to our present settings of interest, Bobkov and
Houdré \cite{BH97} showed that a dimension-independent Poincaré-type inequality also holds for
product measures in $\bR^n$, including the uniform measure on $[0,1]^n$, as shown in
\eqref{eq:intro-poincare-inequality}.

As introduced in the opening, it is these dimension-free Poincaré inequalities for discrete and
continuous product measures whose directed analogues have implications for the structure of monotone
functions and therefore for property testing \cite{GGLRS00,CS16,KMS18}. To enrich the summary laid
out in \cref{table:inequalities}, we present additional related inequalities recently shown by
\cite{BKR20,BCS22,BKKM22} in \cref{table:inequalities-plus}, and briefly explain them here. These
inequalities have unlocked algorithmic results for testing monotonicity of real-valued functions on
the Boolean cube, and Boolean-valued functions on the hypergrid, as summarized in
\cref{sec:related-work}.

\begin{table}[t!]
    \centering
    \begin{NiceTabular}{c | c || c | c | c}[cell-space-limits=0.3em]
        \Block{2-2}{\diagbox{\textbf{Inequality}}{\textbf{Setting} \,}}
            & & \Block{1-2}{\textbf{Discrete}} & & \textbf{Continuous} \\
        & & $\zo^n \to \zo$ & $\zo^n \to \bR$ & $[0,1]^n \to \bR$ \\ \hline \hline

        \Block{2-1}{$(L^1, \ell^1)$-Poincaré}
            & $d^\const_1(f) \lesssim \Ex{\|\grad f\|_1}$
            & * \cite{Tal93} & * \cite{Tal93} & * \cite{BH97} \\ \cline{2-5}
        & $d^\mono_1(f) \lesssim \Ex{\|\grad^- f\|_1}$
            & \cite{GGLRS00} & \cref{thm:main-directed-inequality-discrete}
            & \cref{thm:main-directed-inequality-continuous}
            \\ \hline

        \Block{2-1}{$(L^1, \ell^2)$-Poincaré}
            & $d^\const_1(f) \lesssim \Ex{\|\grad f\|_2}$
            & * \cite{Tal93} & \cite{Tal93} & \cite{BH97} \\ \cline{2-5}
        & $d^\mono_1(f) \lesssim \Ex{\|\grad^- f\|_2}$
            & \cite{KMS18} & ? & \cref{conjecture:better-inequality} \\ \hline \hline

        \Block{3-1}{Related \\ inequalities}
            & $d^\const_0(f) \lesssim \Ex{\|\phi f\|_2}$
            & \Block{2-3}{For $f : \zo^n \to \bR$ \cite{BKR20}} \\ \cline{2-2}
        & $d^\mono_0(f) \lesssim \Ex{\|\phi^- f\|_2}$ && \\ \cline{2-5}
        & $d^\mono_0(f) \lesssim \Ex{\|\phi^- f\|_2}$
            & \Block{1-3}{For $f : [m]^n \to \zo$ \cite{BCS22,BKKM22}} \\
    \end{NiceTabular}
    \caption{Classical and directed functional inequalities on discrete and continuous domains.
    Cells marked with * indicate inequalities that follow from another entry in the table. For
    simplicity, logarithmic factors in the inequalities are ignored.}
    \label{table:inequalities-plus}
\end{table}

Define the vector-valued operators $\phi$ and $\phi^-$ on functions $f : [m]^n \to \bR$ as follows:
for each $x \in [m]^n$ and $i \in [n]$,
\begin{align*}
    (\phi f(x))_i &\define \ind{
        \exists y : (x \preceq_i y \text{ or } y \preceq_i x) \text{ and } f(x) \ne f(y)
    } \,, \\
    (\phi^- f(x))_i &\define \ind{
        (\exists y : x \preceq_i y, f(x) > f(y))
        \text{ or }
        (\exists y : y \preceq_i x, f(y) > f(x))
    } \,,
\end{align*}
where we write $x \preceq_i y$ if $x_j = y_j$ for every $j \ne i$, and $x_i \le y_i$. Compared to
the gradient, these operators 1) are only sensitive to the order relation between function values
(which suits the setting of Hamming testing); and 2) capture ``long range'' violations of
monotonicity (accordingly, the corresponding hypergrid testers are not edge testers).
See \cref{sec:results-inequalities} for a remark on the nuances of inner/outer boundaries and robust
inequalities.

\section{Upper bounds from \cite{BRY14a} applied to Lipschitz functions}
\label{sec:comparison}

In this section, we show how the $L^1$ monotonicity testing upper bounds from \cite{BRY14a} imply
testers with query complexity $\widetilde O\left( \frac{n^2 L}{\epsilon} \right)$ on the unit cube
and $\widetilde O\left( \frac{n^2 mL}{\epsilon} \right)$ on the hypergrid for functions $f$
satisfying $\Lip_1(f) \le L$. We start with the case of the hypergrid, and first state the upper
bound of \cite{BRY14a} for functions with arbitrary range of size $r$, which without loss of
generality (by translation invariance) we denote $[0,r]$:

\begin{theorem}[\cite{BRY14a}]
    \label{thm:bry14a-tester}
    There exists an $L^1$ monotonicity tester for functions $f : [m]^n \to [0,r]$ that uses $O\left(
    \frac{rn}{\epsilon} \log \frac{rn}{\epsilon} \right)$ value queries. The tester is nonadaptive
    and has one-sided error.
\end{theorem}
As explained in \cref{sec:comparison-with-prior-lp}, the extra factor of $r$ compared to the bounds
stated in \cite{BRY14a} accounts for the conversion between range $[0,r]$ and range $[0,1]$, which
affects the proximity parameter $\epsilon$ by a factor of $r$: testing functions with range $[0,r]$
for proximity parameter $\epsilon$ is equivalent to testing functions with range $[0,1]$ for
proximity parameter $\epsilon/r$.

We thus obtain the following $L^1$ monotonicity tester for Lipschitz functions on the hypergrid:

\begin{corollary}
    \label{cor:lp-testing-tester-hypergrid}
    There is an $L^1$ monotonicity tester for functions $f : [m]^n \to \bR$ satisfying $\Lip_1(f)
    \le L$ that uses $O\left( \frac{n^2 mL}{\epsilon} \log\left( \frac{nmL}{\epsilon} \right)
    \right)$ value queries. The tester is nonadaptive and has one-sided error.
\end{corollary}
\begin{proof}
    The algorithm simulates the tester from \cref{thm:bry14a-tester} with parameters $r = Lmn$ and
    $\epsilon$, which gives the announced query complexity. As for correctness, note that since
    $\Lip_1(f) \le L$, for any $x, y \in [m]^n$ we have $\abs*{f(x)-f(y)} \le L\|x-y\|_1 \le Lmn$.
    Thus $f$ has range of size at most $Lmn$, so the reduction to the tester from
    \cref{thm:bry14a-tester} is correct.
\end{proof}

We outline how this result implies a tester for Lipschitz functions on the unit cube as well, via
the domain reduction or downsampling principle of \cite{BCS20,HY22}. Even though the tester above
has query complexity that depends on $m$, the main observation is that, given an $(\ell^1,
L)$-Lipschitz function on $[0,1]^n$, we may discretize it into an $(\ell^1, L/m)$-Lipschitz function
on an arbitrarily fine hypergrid $[m]^n$. Then the term $m(L/m) = L$ remains fixed for any choice of
$m$, so the complexity of the tester above does not depend on $m$ in this reduction. Finally, by
setting $m$ large enough, we may upper bound the error introduced by the discretization.

\begin{corollary}
    There is an $L^1$ monotonicity tester for functions $f : [0,1]^n \to \bR$ satisfying $\Lip_1(f)
    \le L$ that uses $O\left( \frac{n^2 L}{\epsilon} \log\left( \frac{nL}{\epsilon} \right) \right)$
    value queries. The tester is nonadaptive and has one-sided error.
\end{corollary}
\begin{proof}[Proof sketch]
    For sufficiently large $m$ to be chosen below, the algorithm imposes a hypergrid $[m]^n$
    uniformly on $[0,1]^n$. More precisely, we define a function $f' : [m]^n \to \bR$ via $f'(x) =
    f(x/m)$ for all $x \in [m]^n$. Note that $\Lip_1(f') \le L/m$. Then, the algorithm simulates the
    tester from \cref{cor:lp-testing-tester-hypergrid} on $f'$ with proximity parameter
    $\Omega(\epsilon)$, and returns its result. Note that the query complexity is as desired, so it
    remains to show correctness. If $f$ is monotone so is $f'$, in which case the algorithm accepts.
    Now assume that $d_1(f) > \epsilon$, and we need to show that $d_1(f') \gtrsim \epsilon$.

    Let $g' : [m]^n \to \bR$ be any monotone function. We obtain a monotone function $g : [0,1]^n
    \to \bR$ as follows: for each $x \in [0,1]^n$, let $\overline x$ be obtained by rounding up each
    coordinate $x_i$ to a positive integral multiple of $1/m$. Then $x' \define m \overline x \in
    [m]^n$, and we set $g(x) \define g'(x')$. It follows that $g$ is monotone, and by the triangle
    inequality,
    \[
        \Ex{f-g} \le \Ex{f'-g'} + \max_{x \in [0,1]^n} \abs*{f(x) - f(\overline x)}
        \le \Ex{f'-g'} + \frac{Ln}{m} \,.
    \]
    In the first inequality, we separately account for the cost of turning each value of $f$ into
    the value at its corresponding ``rounded up'' point (accounted for by the second summand), and
    then the cost of turning each equal-sized, constant-valued cell into the value of $g$ on that
    cell, and these values agree with $f'$ and $g'$ (accounted for by the first summand). In the
    second inequality, we use the fact that any point $x$ satisfies $\|x - \overline x\|_1 \le
    \frac{1}{m} \cdot n$, along with the Lipschitz assumption on $f$.

    Therefore, by setting $m > \frac{10Ln}{\epsilon}$, we obtain $\frac{Ln}{m} <
    \frac{\epsilon}{10}$, and since the inequality above holds for every monotone function $g'$, we
    conclude that $d_1(f') \ge d_1(f) - \epsilon/10 > 9\epsilon/10$, as desired.
\end{proof}

\begin{remark}
    One may wonder whether the reductions above could yield more efficient testers if combined with
    Hamming testers for Boolean functions on the hypergrid with better dependence on $n$ (via the
    reduction from $L^1$ testing to Hamming testing of \cite{BRY14a}, which is behind
    \cref{thm:bry14a-tester}), since \eg the tester of \cite{BCS23} has query complexity $\widetilde
    O(n^{1/2+o(1)}/\epsilon^2)$. However, it seems like this is not the case, \ie
    \cref{cor:lp-testing-tester-hypergrid} has the best query complexity of any reduction that upper
    bounds the size of the range of $f$ by $Lmn$. The reason is as follows: \cite{KMS18} showed
    that any nonadaptive, one-sided pair tester for the Boolean cube with query complexity
    $O(n^\alpha/\epsilon^\beta)$ must satisfy $\alpha + \frac{\beta}{2} \ge \frac{3}{2}$, so
    hypergrid testers must also satisfy this as well as $\beta \ge 1$, assuming query complexity
    independent of $m$. Then, given a hypergrid tester for Boolean functions with query complexity
    $\widetilde O(n^\alpha/\epsilon^\beta)$, our reduction via the inferred range size $r = Lnm$
    gives an $L^1$ tester with asymptotic query complexity at least
    $\frac{n^\alpha}{(\epsilon/r)^\beta} = \frac{n^{\alpha+\beta/2} n^{\beta/2}
    (mL)^\beta}{\epsilon^\beta} \ge \frac{n^2 mL}{\epsilon}$.
\end{remark}

\section{Lower bound from \cite{BRY14b} applied to $L^1$ testing}
\label{sec:bry14b-lower-bound}

We briefly explain how the $\Omega(n\log m)$ nonadaptive lower bound of \cite{BRY14b} for Hamming
testing monotonicity of functions $f : [m]^n \to \bR$ (with sufficiently small constant $\epsilon$)
also applies to $L^1$-testing functions satisfying $\Lip_1(f) \le O(1)$. The construction of
\cite{BRY14b} relies on two main ingredients: \emph{step functions} and \emph{Walsh functions}.

Let $m = 2^\ell$ for simplicity. For each $i \in \{0, \dotsc, m\}$ the $i$-th step function $s_i :
[2^\ell] \to [2^{\ell-i}]$ is given by
\[
    s_i(x) = \left\lfloor \frac{x-1}{2^i} \right\rfloor + 1 \,.
\]
In words, $s_i(x)$ increases by $1$ after every $2^i$ consecutive elements (called a \emph{block} of
size $2^i$).

The Walsh functions are defined as follows. For each $i \in [\ell]$, the function $w_i : [2^\ell]
\to \pmset$ is given by
\[
    w_i(x) = (-1)^{\mathsf{bit}_i(x-1)} \,,
\]
where the operator $\mathsf{bit}_i$ extracts the $i$-th bit of its input (indexed from least to most
significant). Then for each $S \subseteq [\ell]$, the function $w_S : [2^\ell] \to \pmset$ is given
by
\[
    w_S(x) = \prod_{i \in S} w_i(x) \,.
\]

These two types of functions are defined on the line, and they are extended into a multidimensional
construction on the hypergrid as follows. Given a vector $\bm{i} \in [\ell]^n$, the step function
$s_{\bm{i}} : [2^\ell]^n \to [n 2^{\ell}]$ is given by
\[
    s_{\bm{i}}(x_1, \dotsc, x_n) = \sum_{j=1}^n s_{\bm{i}_j}(x_j) \,,
\]
and given a vector $\bm{S} = (\bm{S}_1, \dotsc, \bm{S}_n)$ of subsets of $[\ell]$, the Walsh
function $w_{\bm{S}} : [2^\ell]^n \to \pmset$ is given by
\[
    w_{\bm{S}}(x_1, \dotsc, x_n) = \prod_{j=1}^n w_{\bm{S}_j}(x_j) \,.
\]

Then, \cite{BRY14b} use a communication complexity argument (namely a reduction from the
\textsc{AugmentIndex} problem) to show that (nonadaptive) Hamming testing monotonicity of functions
$h_{\bm{i},\bm{S}} : [2^\ell]^n \to \bN$ of the form
\[
    h_{\bm{i},\bm{S}}(x) = 2s_{\bm{i}}(x) + w_{\bm{S}}(x) \,,
\]
for appropriate choices of $\bm{i}$ and $\bm{S}$, requires at least $\Omega(n\log m)$ queries.
Therefore, to show that (nonadaptive) $L^1$ testing monotonicity of $(\ell^1, O(1))$-Lipschitz
functions also requires at least this number of queries, it suffices to show that every such
function $h = h_{\bm{i},\bm{S}}$ satisfies
\begin{enumerate}
    \item $\Lip_1(h) \le O(1)$; and
    \item If $d_0(h) \ge \epsilon$, then $d_1(h) \gtrsim \epsilon$.
\end{enumerate}

The first property follows from the definitions of the step and Walsh functions: let $x, y \in
[2^\ell]^n$ be such that $\|x-y\|_1 = 1$. Then let $j \in [\ell]$ be the coordinate such that
$\abs*{x_j - y_j} = 1$ and $x_k = y_k$ for $k \ne j$. Then
\[
    \abs*{h_{\bm{i},\bm{S}}(x) - h_{\bm{i},\bm{S}}(y)}
    \le 2\abs*{s_{\bm{i}_j}(x_j) - s_{\bm{i}_j}(y_j)} + \abs*{w_{\bm{S}}(x) - w_{\bm{S}}(y)}
    \le 2 + 2
    = 4 \,,
\]
the first inequality because $x$ and $y$ agree on every coordinate except for $j$, and the second
inequality because the step function $s_{\bm{i}_j}$ changes by at most $1$ on adjacent inputs, and
the Walsh functions only take values $\pm 1$. Thus $\Lip_1(h_{\bm{i},\bm{S}}) \le 4$, as desired.

As for the second property, note that if $d_0(h) \ge \epsilon$, then there exists a matching of the
form $(x^i, y^i)_i$ where for each $i$ we have $x^i \preceq y^i$ and $h(x^i) > h(y^i)$ (\ie $x^i,
y^i$ form a violating pair), such that at least an $\epsilon$-fraction of the points $[m]^n$ belong
to this matching (see \cite{FLNRRS02}). Now, for each such violating pair $x^i, y^i$, it follows
that $h(x^i) - h(y^i) \ge 1$, since $h$ is integer-valued. Therefore for any monotone function $h' :
[m]^n \to \bR$, it must be the case that $\abs*{h(x^i) - h'(x^i)} + \abs*{h(y^i) - h'(y^i)} \ge 1$.
Since this is true for disjoint pairs $x^i, y^i$ covering an $\epsilon$-fraction of the domain, it
follows that $d_1(h, h') \ge \epsilon / 2$ for any monotone function $h'$. Hence $d_1(h) \ge
\epsilon/2$, and we are done.

\end{document}